\documentclass{amsart}       
\usepackage{graphicx}
\usepackage{mathptmx}

\usepackage{latexsym}
\usepackage{amsmath}
\usepackage{epsfig}
\usepackage{amssymb}
\usepackage{enumerate}
\usepackage{xcolor}

\usepackage{amsmath, amsfonts, amssymb} 
\usepackage{mathtools}
\usepackage{enumitem}

\usepackage[numbers, sort&compress]{natbib}
\usepackage[colorlinks]{hyperref}
\hypersetup{
         citecolor=blue,
    linkcolor=red,
}

\usepackage{color}
\usepackage{xcolor}
\newcommand{\blue}{\textcolor{black}}
\newcommand{\orange}{\textcolor{black}}
\newcommand{\green}{\textcolor{black}}

\usepackage{algorithm}
\usepackage[noend]{algpseudocode}

\newtheorem{theorem}{Theorem}[section]
\newtheorem{lemma}[theorem]{Lemma}
\newtheorem{corollary}[theorem]{Corollary}

\newtheorem{proposition}[theorem]{Proposition}

\begin{document}

\title{The hybrid number of a ploidy profile}

\author{Katharina T. Huber \and Liam J. Maher}
         \address{School of Computing Sciences, University of East Anglia,
Norwich, UK}
   \email{K.Huber@uea.ac.uk}
         \address{School of Computing Sciences, University of East Anglia,
Norwich, UK }
  \email{L.Maher@uea.ac.uk}


\date{\today}

\maketitle

\begin{abstract}
Polyploidization, whereby an organism inherits multiple copies of the genome of their parents, is an important evolutionary event that has been observed in plants and animals. One way to \green{study such events}
is in terms of the ploidy number of the species that make up a dataset
of interest.
It is therefore natural to ask: How much information about the evolutionary past 
of \green{the set of species that form} a dataset can be gleaned from the ploidy numbers of 
the species? 
To help answer \green{this} question, we introduce and study 
the novel concept of a ploidy
profile which allows us to formalize \green{it}
in terms of a multiplicity vector indexed by the species
the dataset \green{is comprised of}. Using the framework of a
phylogenetic network, \green{we present} a 
closed formula for computing \green{ the {\em hybrid
number} (i.e. the minimal number of polyploidization
events required to explain a ploidy profile) of a large class of ploidy profiles. This formula relies on the construction of a certain phylogenetic network from 
the simplification sequence of a ploidy 
profile and the hybrid number
of the ploidy profile with which this construction is
initialized. Both of them can be computed easily in case the ploidy numbers that make up the ploidy profile are not too large. To help illustrate the applicability of our approach, we apply it to a simplified   version of a publicly available \orange{Viola} dataset.}

\end{abstract}

\section{Introduction}
\label{intro}

Datasets such as the \blue{Viola} dataset considered in
\blue{\cite{MJDBBBO12}}  arise when species  inherit 
multiple sets of chromosomes from their parents. 
\green{Generally referred to as polyploidization,
this can be due to whole genome duplication \blue{(also called autopolyplodization)}
as in the case of e.g. watermelons and bananas  \cite{VBDG00},
or by obtaining an additional complete set of chromosomes
via hybridization \blue{(also called allopolyploidization)}, as in the case of the frog 
genus Xenopus
 \cite{O50}. This poses the following intriguing
  question at the center of this paper: How much 
  information about the evolutionary past of a
   set of species can be gleaned from the {\em ploidy number} (i.e. the number of complete chromosome sets in a genome)
   of the species? }  Evoking
   parsimony to capture the idea that polyploidization is a relatively 
   rare evolutionary event we \green{re-phrase} this question as follows: What is the 
   minimum number of polyploidization events
   necessary to explain a dataset's observed 
   \green{{\em ploidy profile}. For a set $X$ of species that make up a dataset, we define such a profile to be 
   	the multiplicity
   	vector $(m_1,\ldots, m_n)$ for $n=|X|$, indexed by the species in $X$ where, for each $1\leq i\leq n$, the ploidy number 
   	of species $i\in X$ is $m_i\geq 1$.}
   
   As it turns out, an answer to this question is
    well-known if the ploidy profile in question
    is presented in terms of 
   a multi-labelled tree (see e.g. \cite{HM06,HOLM06,MHBOJ15,MJDBBBO12}). 
   Since it is, however, not always clear how
   to derive \blue{a biologically meaningful multi-labelled} \orange{tree} from the dataset in the first
   place \cite{HMSSS12}, we focus here on
  ploidy profiles for which such a 
  tree is not necessarily available. 
   
Due to the reticulate nature of the signal left 
behind by polyploidization \cite{JSO13,OTK17,BTWKW}, phylogenetic
 networks offer themselves 
 as a natural framework to formalize and answer our question.
 \blue{Although we present a definition of such structures (and
 	all other concepts used in this section) below, from an intuition development point of view, it suffices to observe at this stage that a phylogenetic network can sometimes be thought of as a rooted directed bifurcating tree $T$ with a pre-given set $X$ as leaves 
 	to which additional arcs have been added via joining subdivision vertices of arcs of $T$ so that
 	the following property holds. The resulting graph is a rooted directed acyclic graph with leaf set $X$ such that 
 	a subdivision vertex $v$ of $T$ either only has additional arcs
 	starting at it or only additional arcs ending at it. For our 
 	purposes we only allow the case that $v$ has one
 	additional outgoing arc. Subdivision vertices that have at least one additional incoming arc are called  {\em hybrid vertices} and are assumed to represent reticulate evolutionary
 	events such as polyploidization. If a hybrid vertex \blue{in a phylogenetic network $N$} also has 
 	overall degree three then $N$ is generally called a {\em binary} phylogenetic network. We refer the interested reader to
 	 Figure~\ref{fig:intro-fig}(i) for an example of a binary phylogenetic network on
 	 $X=\{x_1,x_2,x_3,x_4\}$ that is obtained from the tree depicted in Figure~\ref{fig:intro-fig}(ii)  and to
	\cite{G14,HM13,HRS10,S16}
for methodology and 
construction algorithms surrounding phylogenetic networks.} Note that 
to be able to account for 
autopolyploidization, we deviate from the usual notion of a phylogenetic network by allowing our \blue{phylogenetic} networks to
have parallel arcs \blue{(but no loops) -- see} e.g. \cite{HLM20,IJJMZ19} and the references therein for further results concerning such
networks. 

By taking for every leaf $x$ of a binary phylogenetic 
network $N$ on some finite set $X$ the number of directed 
paths from the root of $N$ to $x$, every phylogenetic network
induces a multiplicity vector $\vec{m}$ indexed by the elements in $X$. Saying that $N$ {\em realizes}
	$\vec{m}$ in this case (see Section~\ref{sec:structural} for \blue{an extension} of this concept \blue{to phylogenetic networks}) allows us to formalize our
 question as follows. Suppose $\vec{m}$ is a ploidy profile indexed by the elements of some finite 
set $X$. What can be said about the minimum number of hybrid vertices required by a binary phylogenetic network on
$X$ to realize $\vec{m}$? We \blue{call} this number which is central to the paper the {\em hybrid number} of $\vec{m}$ and denote it by $h(\vec{m})$. \blue{If a binary phylogenetic network $N$ has $h(\vec{m})$ hybrid vertices then we also say that $N$ {\em attains} 
$\vec{m}$} \blue{(see again Section~\ref{sec:structural} for an extension of this concept to phylogenetic networks)}. The interested reader is referred to \cite{S16} for an overview \blue{of} the related concept of the hybrid number of a set
of phylogenetic trees (i.e. leaf-labelled rooted trees without any vertices of
indegree and outdegree one \blue{whose leaf set is a pre-given set).} 

Before proceeding with presenting an example to help illustrate this question we remark that multiplicity vectors realized by binary phylogenetic networks have been used in \cite{CLRV08} to define a metric for \blue{a certain class of binary} phylogenetic networks.
	Furthermore, the stronger assumption that the 
	number of directed paths from {\em every} vertex of a \blue{binary} phylogenetic network $N$
	to every leaf of $N$ is known, \green{has led to the introduction of the concept of} an ancestral 
	profile for $N$ \cite{ESS19}.

\green{Returning to our question, consider} the ploidy profile $\vec{m}=(\green{12}, 6, 6, 5)$ indexed by $X=\{x_1,x_2,x_3,x_4\}$
where the multiplicity of $x_1$ is $12$, that of
$x_2$ and $x_3$ is 6, and that of $x_4$ is $5$. Since no \blue{binary} phylogenetic 
	network on one leaf and two hybrid vertices can realize the
	ploidy profile $\vec{m}'=(5)$ because \blue{it has at most} $2^2=4$ \blue{directed
	paths from the root to the leaf}, it follows that a \blue{binary}  phylogenetic network that
	realizes $\vec{m}'$ and therefore also $\vec{m}$ must have at least three hybrid vertices. In fact, the subnetwork \blue{$N'$} in bold of the phylogenetic network depicted
	in \blue{Figure~\ref{fig:intro-fig}(i)} is the unique  (subject to 
	letting the arc $a$ finish at a subdivision vertex of an outgoing or incoming arc of the hybrid vertex $h$ or letting $a$ start at a subdivision vertex of an outgoing or incoming arc of the vertex $t$) \blue{binary} 
	phylogenetic network that realizes $\vec{m}'$ and
	uses a minimum number of hybrid vertices. 
	To be able to realize the ploidy profile $(6,5)$ and therefore also the ploidy profile 
	$\vec{m}''=(6,6,5)$ at least four hybrid vertices are therefore needed. \blue{By counting directed paths from the root to each leaf of} the \blue{phylogenetic} network \blue{depicted} in \blue{Figure~\ref{fig:intro-fig}(i)} with $x_1$, the hybrid vertex $h'$ above $x_1$, the two incoming arcs of $h'$, and the arc
	$(h',x_1)$ removed and any resulting vertices of indegree and outdegree one suppressed \blue{clearly} realizes $\vec{m}''$.
	Calling that \blue{phylogenetic} network $N''$ then, in a similar sense as
	$N'$, we also have that $N''$ is unique. To obtain a \blue{binary phylogenetic} network from $N''$ that realizes $\vec{m}$ at least one
	further hybrid vertex is needed.
\blue{Again by counting directed paths from the root to each leaf}, it is easy to check that the \blue{binary phylogenetic network} $N(\vec{m})$ depicted in \blue{Figure~\ref{fig:intro-fig}(i)}
\begin{figure}[h]
	\centering
	\includegraphics[scale=0.25]{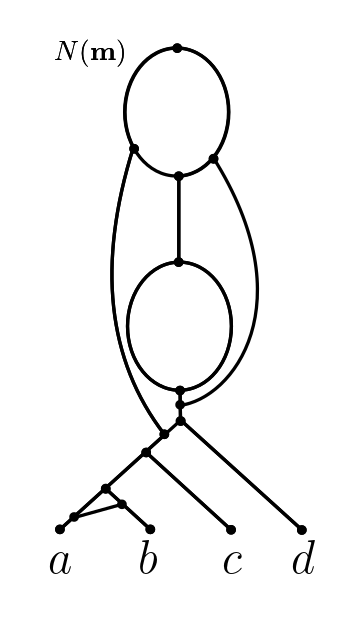}
	\caption{\label{fig:intro-fig}
		\blue{(i)} One of potentially many phylogenetic \green{networks} that realize the 
		ploidy profile $\vec{m}=(12, 6, 6, 5)$ 
		on $X=\{x_1,x_2,x_3,x_4\}$. 
		To improve clarity of exposition, we always assume that arcs are directed \green{downward}, away from the root. \blue{\orange{(ii)} A (phylogenetic) tree to which subdivision vertices and arcs have been added to obtain the phylogenetic network in (i) -- see the text for details.}
		}	
	\end{figure}
realizes $\vec{m}$ and 
postulates \green{five} hybrid  vertices. 
 As we shall see \orange{as a direct consequence
 of Theorem~\ref{theo:min}, $h(N)=5$.} \blue{As a further consequence of that theorem, we obtain a
 closed formula for the hybrid number of a ploidy profile (Corollary~\ref{cor:upper-bound}).}

The outline of the paper is as follows. 
\green{In the next section, we present}
some relevant basic terminology and 
notation concerning phylogenetic networks. This also
includes an unfold-operation for phylogenetic networks and a
fold-up \blue{operation} that generates phylogenetic networks, both of which  
were introduced originally in \cite{HM06}.  
In Section~\ref{sec:structural}, we \blue{extend the concept
of \orange{attainment} from binary phylogenetic networks to phylogenetic networks and}
study structural properties of phylogenetic networks that attain 
ploidy profile. As part of this, we
introduce the two main concepts of the paper:
	a simple ploidy profile and an attainment of a ploidy profile. In Section~\ref{sec:simple},
	 we associate two \blue{binary} phylogenetic networks
to a simple ploidy profile $\vec{m}$ which we denote by
\blue{$D(\vec{m})$ and $B(\vec{m})$}, respectively. 
As we shall see, the former is based on the prime factor decomposition
of a positive integer $m$
and the latter on a binary representation of $m$. 

In Section~\ref{sec:n(m)}, we associate \blue{a sequence $\sigma(\vec{m})$}
to a ploidy profile $\vec{m}$ 
\blue{which we call the simplification sequence of $\vec{m}$
\blue{(Algorithm~\ref{alg:simplification-seq})}. As part of this, we
also}  present some
basic results concerning such sequences. This includes
an infinite family of
ploidy profiles that shows that such a sequence
can grow exponentially large. Denoting the last element  
of \blue{the} simplification sequence for $\vec{m}$ by $\vec{m}_t$, we then employ 
a traceback through
$\sigma(\vec{m})$ to obtain \blue{the aforementioned binary phylogenetic} network $N(\vec{m})$
from  \blue{a binary phylogenetic network that attains} $\vec{m}_t$ (Algorithm~\ref{alg:Nmalgorithm}). 
Motivated by our partial results for \blue{binary phylogenetic networks that realize a}  simple ploidy profile \blue{summarized in Theorem~\ref{Bm-vertex-count}}, we provide 
an upper bound on the hybrid number $h(\vec{m})$  of
a ploidy profile $\vec{m}$ 
for special cases of $\vec{m}$ (Proposition~\ref{prop:upper-bound}).

After collecting some
preliminary results for $N(\vec{m})$ in Section~\ref{sec:n(m)},
we establish in Section~\ref{sec:optimal-n(m)}
that $N(\vec{m})$ \blue{attains $\vec{m}$}
for a large class of ploidy profiles $\vec{m}$
(Theorem~\ref{theo:min}). 
 In Section~\ref{sec:viola}, we turn our attention to computing the hybrid number 
 of the ploidy profile of \blue{a simplified version of the aforementioned} \orange{Viola} dataset \blue{from} \cite{MJDBBBO12}.
We conclude with Section~\ref{sec:discussion} where we outline potential directions of further research.



\section{Preliminaries}
\label{sec: preliminaries}
We start with introducing basic concepts surrounding 
phylogenetic networks. \green{Subsequent to this, we} briefly describe two 
basic operations concerning phylogenetic networks
\green{that are central for establishing a key result (Proposition~\ref{prop:key}).}
For the convenience of the reader, we illustrate \green{both operations} in 
Figures~\ref{fig:fun-construct-idea} and \ref{fig:fun-construct} 
by means of an example. Throughout the paper we assume that $X$ is a non-empty
finite set. We denote the size of $X$ by $n$. 

\subsection{Basic concepts}
Suppose for the following that $G$ is a rooted directed connected acyclic graph which might
contain parallel arcs \blue{but no loops}. Then we denote
the vertex set of $G$ by $V(G)$ and its set of arcs
by $A(G)$. We denote an arc $a\in A(G)$
starting at a vertex $u$ 
and ending in a vertex $v$ by $(u,v)$ and refer to 
$u$ as the {\em tail} of $a$ and to $v$ as the {\em head}
of $a$. We call an arc $a\in A(G)$ a {\em cut-arc} if the deletion
of $a$ \blue{disconnects $G$.} We call
a cut-arc $a$ of $G$ {\em trivial} if the head of $a$ is a leaf. 
\orange{Following \cite{IJJMZ19}, we call an induced subgraph of $G$ with
	two vertices $u$ and $v$ and two parallel arcs form $u$ to $v$ a {\em bead} of $G$.} 

Suppose $v\in V(G)$. Then  
we refer to the number of arcs coming into $v$
as the {\em indegree} of $v$, denoted by $indeg_G(v)$,
and the number of outgoing arcs of $v$ as the
{\em outdegree} of $v$, denoted by $outdeg_G(v)$.
If $G$ is clear from the context then we will
omit the subscript in $indeg_G(v)$ and
$outdeg_G(v)$, respectively. We call $v$ the {\em root} 
of $G$, denoted by $\rho_G$,
if $indeg(v)=0$,  and we call $v$ a {\em leaf} 
of $G$ if $indeg(v)=1$ and $outdeg(v)=0$. We denote 
the set of leaves of $G$ by $L(G)$.
We call $v$
a {\em tree vertex} if $outdeg(v)=2$ and $indeg(v)=1$.
And we call $v$ a {\em hybrid vertex} if $indeg(v)\geq2$ and $outdeg(v)=1$. 
We denote the set of hybrid vertices of $G$ by $H(G)$.
\green{We call any two leaves $x$ and $y$ of $G$
	a {\em cherry}, denoted by $\{x,y\}$, if $x$ and $y$ share a parent.}
We say that $G$ is {\em binary} if,
$outdeg(\rho_G)=2$ and, 
for all $v\in V(G)-L(G)$ other
than $\rho_G$, we have that the degree sum is three.
We say that a vertex $w\in V(G)$ is {\em above}
$v$ if there exists a directed path $P$ 
from $w$ to $v$. In that case, we also say 
that $v$ is {\em below} $w$. \green{If, in addition, $v\not=w$}  
	then we say that $w$ is {\em strictly above} $v$
and that $v$ is {\em strictly below} $w$. 

We call $G$ a {\em (phylogenetic) network (on $X$)}
if $L(G)=X$, every vertex 
$v\in V(G)-L(G)$ other than $\rho_G$ is a tree vertex
or a hybrid vertex and
\blue{$outdeg(\rho_G)= 2$}.  Note that phylogenetic networks in our sense were called semi-resolved
phylogenetic networks in \cite{HM06}. Also note that
our definition of a phylogenetic network differs 
from the standard definition of such an object (see e.g. \cite{S16}) by allowing
\blue{beads. To emphasise that a phylogenetic network has no
	beads, we will sometimes refer to it as a {\em beadless} phylogenetic network.}  

Suppose $G$ is a phylogenetic 
network on $X$. Then following \cite{BS07},
we define the {\em hybrid number} $h(G)$ of $G$
to be 
$$
h(G)=\sum_{h\in H(G)} (indeg(h)-1).
$$ 
We refer to a phylogenetic network $G$
(on $X$) as a {\em phylogenetic tree (on $X$)} if
$h(G)=0$. For a phylogenetic tree $T$ on $X$
and a non-root vertex 
$v\in V(T)$ we denote by $T(v)$ the subtree of $T$ obtained by deleting the incoming
arc of $v$ and \blue{the subsequently generated connected component that does not contain $v$.}

Suppose that $N$ is a 
phylogenetic network on $X$. Then we denote the number of
directed paths from the root $\rho_N$ of $N$ to a leaf $x$
of $N$ by $m_N(x)$. In case $N$ is clear from the context,
we will write $m(x)$ rather than  $m_N(x)$. For $N'$
a further phylogenetic network on $X$, we say that
$N$ and $N'$ are {\em equivalent} if there exists a
graph isomorphism between $N$ and $N'$ that is the
identity on $X$. Furthermore, we say that $N'$ is
a {\em \blue{(binary)} resolution} of $N$ if $N'$ is 
obtained from $N$ by resolving all vertices in $H(N)$ so that every 
vertex in $H(N')$ has indegree two. 
%
Note that for any resolution $N'$ of $N$, we have $h(N)=|H(N')|=h(N')$. 

\subsection{\green {The fold-up $F(U(N))$ of the unfold $U(N)$ of a phylogenetic network $N$}}
\label{section:fun}

Phylogenetic trees on $X$ were generalized in \cite{HM06} to  so called  {\em multi-labelled trees (on $X$)} or {\em MUL-trees (on $X$)}, for short, by replacing the leaf set of 
a phylogenetic tree  by a multiset $Y$ on $X$. Put differently, $X$ is the set obtained from $Y$ by ignoring \blue{the} multiplicities
of the elements in $Y$. As was pointed out in the same paper, 
every phylogenetic network $N$ gives rise to a 
\blue{MUL-tree} $U(N)$ on $X$ by recording, for every vertex $v$ of $N$, every directed 
path from the root $\rho_N$ of $N$ to $v$. More precisely, the vertex set
of $U(N)$ is, \green{for all vertices $v\in V(N)$,} the set of all directed paths $P$ from $\rho_N$ to $v$ where we identify $P$ with its end vertex $v$. Two
vertices $P$ and $P'$ in $U(N)$ 
are joined by an arc $(P',P)$ if there exists an arc $a\in A(N)$ such that $P$ is 
obtained from $P'$ by extending $P'$ by the arc $a$.
For example, the vertex $u$ in Figure~\ref{fig:fun-construct-idea}(i) 
	\begin{figure}[h]
		\centering
		\includegraphics[scale=0.33]{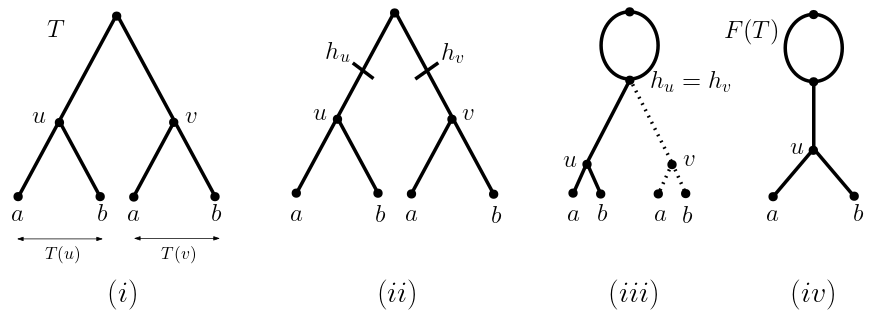}
		\caption{\label{fig:fun-construct-idea}
			(i) The MUL-tree $M$ obtained by unfolding the 
			phylogenetic network on $X=\{x,y\}$ in (iv). 
			\blue{The trees} $T(u)$ and $T(v)$ rooted at \blue{$u$ and $v$} and indicated with a double arrow, \blue{respectively,
			are equivalent. In fact, they are maximal inextendible.}
			(ii) Subdivision of the incoming arcs of $u$ and $v$ by $h_u$ and $h_v$, respectively. (iii) Identifying the vertices 
			$h_u$ and $h_v$. \blue{(iv)} Deleting the \green{subtree} $T(v)$
			and the incoming arc of $v$ (indicated by dotted lines \blue{in (iii)}). 
		}	
	\end{figure}
	is the directed path $\rho$, $s$, $u$ in the phylogenetic network in 
	Figure~\ref{fig:fun-construct-idea}(iv)
	which crosses the arc $a$. The vertex $v$ in Figure~\ref{fig:fun-construct-idea}(i) is the directed path $\rho$, $s$, $u$ in Figure~\ref{fig:fun-construct-idea}(iv) which crosses the arc $a'$.

\blue{Reading Figure~\ref{fig:fun-construct-idea} from left to right suggests that} the \blue{unfolding operation} can also be \orange{reversed.} \blue{We} next briefly outline \blue{this reversal operation which may be thought of as} the fold-up \blue{of} a MUL-tree $M$ into a phylogenetic network $F(M)$
(see \cite{HM06} for details, \cite{HMSW16,HS20}
for more on both constructions, and Figure~\ref{fig:fun-construct} for an example).
\blue{To make this more precise, we require further} terminology. 
Suppose that $M$ is a \blue{MUL-tree on $X$. Then we denote for
	 a non-root vertex $v$ of $M$ the parent of $v$ by $\overline{v}$.
	 Extending the relevant notions from phylogenetic trees to MUL-trees, 
	 we say that a subMUL-tree $T$ with root $u$ of $M$ is 
	{\em inextendible} if there exists a subMUL-tree $T'$ of $M$ with root vertex $w\not=u$
such that  $T$ and $T'$ are equivalent and either $\overline{v}=\overline{w}$
or $\overline{v}\not=\overline{w}$ and $T(\overline{v})$ and $T(\overline{w})$ 
are not equivalent. By definition, 
every subMUL-tree of $M$ that is equivalent with an
inextendible subMUL-tree of $M$ is necessarily also inextendible. In view of this,
we refer to an inextendible subMUL-tree $T$ of $M$ as {\em  maximal inextendible} if 
no subMUL-tree of $M$ that is equivalent with $T$ is a subMUL-tree of an inextendible subMUL-tree of $M$.	
So, for example, the subMUL-tree $T(u)$ of the MUL-tree $M$ \blue{depicted} in Figure~\ref{fig:fun-construct}(i) 
is \orange{inextendible} but the subMUL-tree $T(u')$ is not. In fact, $T(u)$ is maximal inextendible because the only equivalent copy of $T(u)$ in $M$ that is not $T(u)$
is $T(v)$ and neither $T(u)$ nor $T(v)$ is a subMUL-tree of an inextendible subMUL-tree in $M$. } 

\blue{To \orange{construct} $F(M)$, we first construct a sequence 
	$\gamma_M$ of subMUL-trees of $M$ which we call a {\em guide sequence} for 
	$F(M)$ and which we initialize with the empty sequence. 
	Let $T$ denote a maximal inextendible subMUL-tree of $M$. Let $u$ denote the root of $T$, and let  $\blue{U=U_u}\subseteq V(M)$ denote the
set of vertices $v\in V(M)$ such that the 
subMUL-tree rooted at $v$ is equivalent with} $T(u)$.  Note that, by definition, $|U|\geq 2$. Then, for all \blue{$v\in U$}, we first subdivide the incoming arc of \blue{$v$} by a vertex 
\blue{$h_v$ (cf Figure~\ref{fig:fun-construct-idea}(ii))} and then identify all
vertices \blue{$h_v$, $v\in U$,  with the  vertex $h_u$ (cf Figure~\ref{fig:fun-construct-idea}(iii)). By construction, $h_u$ clearly has $|U|$ incoming arcs} and also $|U|$ outgoing arcs. From these $|U|$ outgoing arcs of $h_u$, \blue{we} delete all but one arc and, for each
deleted arc $a$, \blue{we} remove the subMULtree $T(v)$ rooted at the head $v$ of $a$  \blue{(Figure~\ref{fig:fun-construct-idea}(iv))}.
We then grow $\gamma_M$ by 
adding an equivalent copy of $T(u)$ at the 
end of $\gamma_M$ in case $\gamma_M$ is not the empty sequence.
Otherwise we add $T(u)$ as the first element to 
$\gamma_M$. Replacing 
$M$ with the resulting graph $N_U$, we then find a new \blue{maximal inextendible
	\orange{subMUL-tree} in $N_U$} and proceed as before (where we canonically extend the notions of a \blue{maximal inextendible subMUL-tree and} of a subMUL-tree rooted at a vertex to $N_U$). In the case of the example in Figure~\ref{fig:fun-construct}, the next \blue{maximal inextendible subMUL-tree} in Figure~\ref{fig:fun-construct}(ii) is \blue{one of the leaves labelled $x_1$.}

\blue{By construction,} the process of subdividing \blue{(cf Figure~\ref{fig:fun-construct-idea}(ii))},
identifying \blue{(cf Figure~\ref{fig:fun-construct-idea}(iii)))}, and deleting \blue{(cf Figure~\ref{fig:fun-construct-idea}(iv))} 
terminates in a phylogenetic network on 
$X$. \blue{That network is $F(M)$}. We depict $F(M)$ in Figure~\ref{fig:fun-construct}(\blue{iv}) for the MUL-tree \blue{$M$} pictured in
Figure~\ref{fig:fun-construct}(i).
 
As was \blue{pointed out in \cite[Section 6]{HM06}}, $F(M)$ is
independent of the order in which ties are resolved when processing 
	\blue{maximal inextendible subMUL-trees}. Also, all tree vertices of $F(M)$
have outdegree two because $M$ is a binary MUL-tree. However,
$F(M)$ might contain hybrid vertices whose indegree is two or more since when processing a maximal inextendible subMUL-tree $T$  there might be more than two subMUL-trees in the graph generated thus far that are equivalent with $T$. \blue{Finally, $F(M)$ cannot contain arcs whose tail and head is a hybrid vertex
	because the hybrid vertices of $F(M)$ are in bijective correspondence with the
	elements in the guide sequence for $F(M)$.}

\begin{figure}[h]
	\centering
	\includegraphics[scale=0.43]{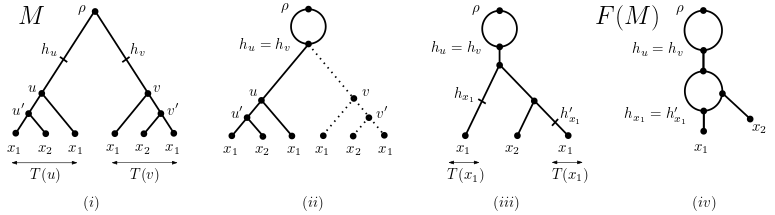}
\caption{\label{fig:fun-construct}
	(i) The MUL-tree $M$ obtained by unfolding the phylogenetic 
	network on $\{x_1,x_2\}$ pictured in (iv). The vertices $u$ and $v$ as indicated in (i) \blue{are} the 
	\blue{root of the maximal inexendible subtrees of $M$}
	 to which the subdivision, identification and deletion process described in \blue{Figure~\ref{fig:fun-construct-idea}} is applied to obtain the rooted directed acyclic graph \blue{$G$} presented in (iii).  The two leaves labelled $x_1$ in  \blue{$G$ are the roots of \blue{two equivalent} maximal inextendible subtree of $G$} and applying the subdivision,
	identification, and deletion process to it results in $F(M)$.
	In each case, the equivalent subMUL-trees are indicated by a double arrow.	
	}
\end{figure}

We conclude the outline of both constructions with the following remark. Suppose $N$ is a phylogenetic network on $X$. \blue{Then we call two tree vertices $u$ and $v$ 
	in $V(N)$ distinct an {\em identifiable pair} if the subMUL-trees of $U(N)$ rooted at 
	\blue{the vertex that is} a directed path
in $N$ from the root $\rho_N$ of $N$ to $u$ is equivalent with the subMUL-trees of $U(N)$ rooted at \blue{the vertex that is} a directed path
in $N$ from $\rho_N$  to $v$.} Let $C(N)$ denote the {\em compressed} phylogenetic network obtained from $N$  \blue{i.\,e.\,the phylogenetic network obtained from $N$ by contracting} all arcs $(u,v)$ for which both $u$ and $v$ is a hybrid vertex. Bearing in mind that the phylogenetic network $F(M)$
associated to a MUL-tree $M$ was denoted $\mathcal D(M)$ in \cite{HM06}, \blue{the following holds
	\begin{itemize}
		\item[(R1)] $F(U(N))$ does not contain an identifiable pair of vertices
		\cite[Theorem 3]{HM06}. 
		\item[(R2)] If $N$ and $N'$ are phylogenetic networks such
		that the MUL-trees $U(N)$ and $U(N')$ 
		are equivalent then $h(F(U(N)))\leq h(N')$ \cite[Corollary 2(ii)]{HM06}.
		\item[(R3)] If $N$ is a phylogenetic network that does not contain an
		identifiable pair of vertices then the compressed phylogenetic networks $C(F(U(N)))=F(U(N))$ and $C(N)$ are equivalent
		(Consequence of (R1) and \cite[Theorem 2]{HM06}).
		\end{itemize}
	}
	%

\section{Properties of phylogenetic networks that attain the hybrid number of a ploidy profile}
\label{sec:structural}

In this section, we collect structural properties of
phylogenetic networks that attain
the hybrid number of a ploidy profile.  
\green{For ease of readability, we will assume 
from now on that for a ploidy profile
$\vec{m}=(m_1,\ldots, m_n)$ on  $X$
the elements in $X$ are always ordered in
such a way that $m(x_i)=m_i$ holds for all $ 1\leq i\leq n$
and that $\vec{m}$ is in {\em descending order},
that is,  $m_i\geq m_{i+1}$ holds for all $1\leq i\leq n-1$.}

 We start with some notations and definitions. 
Suppose \blue{that} $N$ is a phylogenetic
network on $X=\{x_1,\ldots, x_n\}$
\blue{and that} $\vec{m}=(m_1,\ldots, m_n)$ 
	is a ploidy profile on $X$. Then we call $\vec{m}$
	{\em simple} if
	\blue{$m_i = 1$} for all $2\leq i\leq n$ (i.\,e.\, $m_1$ is 
	the only component of $\vec{m}$ that is at least \blue{$2$}).
	Moreover, we call $\vec{m}$
{\em strictly simple} if $\vec{m}$ is simple and 
	$|X|=1$. We
say that $N$ {\em realizes} a ploidy profile
$\vec{m}$ if 
 the elements in $X$ can be ordered in
such a way that $m_i=m(x_i)$ holds for all $1\leq i\leq n$.
\blue{In this case, we also call $N$ a {\em realization} of $\vec{m}$.
	Furthermore, we say that $N$ is a {\em binary} realization of
	$\vec{m}$ if $N$ is binary.	}
We say that $N$ {\em attains}  $\vec{m}$ if $N$ realizes
$\vec{m}$ and $ h(\vec{m})=h(N)=\sum_{h\in H(N)} (indeg(h)-1)$. In this case, we refer to $N$
as an {\em attainment} of $\vec{m}$. \blue{If $N$ is an attainment and also binary
	then we call $N$ a  {\em binary} attainment of $\vec{m}$.}

As is straight-forward to verify using the construction
	of the phylogenetic network indicated in Figure~\ref{fig:maxhybridmul}
	\blue{and the definition of $m(x)$, $x\in X$},
	every ploidy profile $\vec{m}=(m_1,\ldots, m_n)$ on $X=\{x_1,\ldots, x_n\}$ with $n\geq 1$ is realized by a phylogenetic network
	that contains at most $ \sum_{i=1}^n(m_i-1)$  hybrid vertices.
	Thus, the hybrid number of a ploidy profile always exists.
\begin{figure}[h]
	\centering
	\includegraphics[scale=0.3]{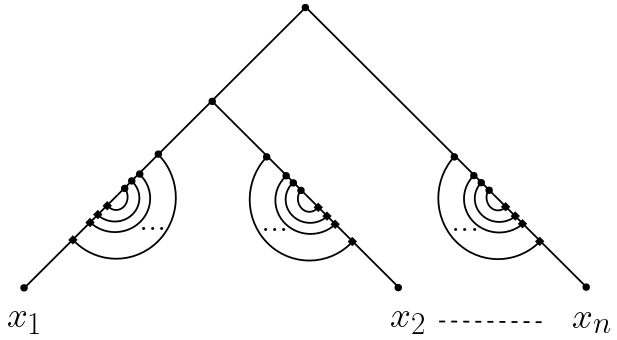}
	\caption{ \green{A phylogenetic network
		on $X=\{x_1,\ldots, x_n\}$ that realizes the ploidy profile
		$\vec{m}=(m_1,\ldots, m_n)$ on $X$. For all 
		$1\leq i\leq n$, \green{the number of curved lines}
		is $m_i-1$.}
		\label{fig:maxhybridmul}
	}
\end{figure}
As we shall see in Proposition~\ref{prop:upper-bound}, this bound can be improved for many ploidy profiles. 

To be able to collect some simple properties of attainments which we will do next, we require further terminology and notation. Suppose $N$ is a \blue{binary} phylogenetic network on $X$.
	Then we say that $N$ is {\em semi-stable} 
		if $N$ is equivalent \blue{to} a resolution of $F(U(N))$.
		 \blue{Motivated by the fact that a beadless phylogenetic
	network $N$ that is equivalent to} $F(U(N))$ \blue{was} called {\em stable} \blue{in} \cite{HMSW16}, we canonically extend this concept to our types of phylogenetic networks \blue{by saying that a phylogenetic network $N$ is {\em stable} if
$N$ is equivalent with $F(U(N))$}. 

\blue{For example, the binary phylogenetic network $N$ depicted in 
	Figure~\ref{fig:semibinaryforbid}(i) is semi-stable but not stable since $U(N)$
	is the MUL-tree depicted in  Figure~\ref{fig:semibinaryforbid}(ii) and $F(U(N))$
	is the phylogenetic network depicted in Figure~\ref{fig:semibinaryforbid}(iii). The
	phylogenetic network $N'$ pictured in Figure~\ref{fig:semibinaryforbid}(iv) is not
	semi-stable. In fact, for a binary phylogenetic network $N$ to be stable it cannot contain the phylogenetic network $N'$ pictured  in Figure~\ref{fig:semibinaryforbid}(iv) as an induced subgraph (where $x_1$  and $x_2$ need not be leaves in $N'$) since $F(U(N'))$ is the phylogenetic network depicted in Figure~\ref{fig:semibinaryforbid}(v).}
\begin{figure}[h]
	\centering
	\includegraphics[scale=0.34]{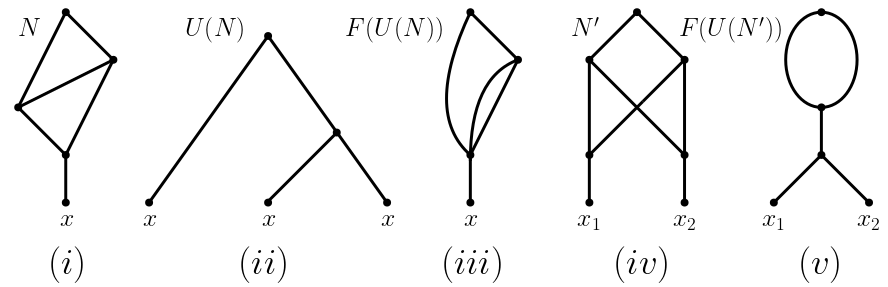}
	\caption{The phylogenetic network $N$ depicted in (i) is semi-stable but not stable since it is not equivalent with $F(U(N))$ i.\,e.\,the phylogenetic network depicted in (iii). the MUL-tree $U(N)$ is pictured in (ii). The phylogenetic network pictured in (iv) is not semi-stable. For a phylogenetic network to be stable it cannot contain the phylogenetic network $N'$ pictured in (iv) as an induced subgraph since  $F(U(N'))$.		 
		\label{fig:semibinaryforbid}
	}
\end{figure}
\blue{As we shall see below, certain types of binary phylogenetic networks called} beaded trees \blue{are examples of stable phylogenetic networks. Although introduced} in  \cite{IJJMZ19}
\blue{in the context of a study of binary phylogenetic networks whose root have indegree one and not zero as in our case, \orange{the main feature of beaded trees} is that 
	a hybrid vertex must be contained in a bead.  In view of this, we call a 
	binary phylogenetic network $N$ on $X$ a
	{\em beaded tree} if $N$ is either a phylogenetic tree on $X$ or
	every hybrid vertex is contained in a bead} (see e.\,g.\,\cite{HLM20} for more on such graphs). \blue{Then since a 
	beaded tree $N$ cannot contain an identifiable pair of vertices, it follows by (R3) that the compressed phylogenetic networks $C(N)$
	and $F(U(N))$ are equivalent. Since $N$ is a beaded tree and so does not contain arcs whose tail and head are hybrid vertices, it follows that $C(N)$ is in fact $N$. Thus, $N$ must be stable.}
	
Suppose $N$ is an attainment of a ploidy profile $\vec{m}$ on $X$ that contains a cut-arc $a$. Then deleting $a$ results in two connected components $N_1$ and $N_2$, one of which contains the root of $N$, say $N_1$, and the other is a phylogenetic network on $X-L(N_1)$. For $x\not\in L(N_1)$ we let $N_1^x$ denote the phylogenetic network on $L(N_1)\cup \{x\}$ obtained from $N_1$ by adding a pendant arc $a'$ to 
	$tail(a)$ and labelling the head of $a'$ by $x$.
	For any phylogenetic network $N$ on $X$, 
	we denote by $\vec{m}(N)$
	the ploidy profile on $X$ realized by $N$.

\begin{lemma} \label{lem:folding}
	Suppose that $N$ is an attainment of a ploidy profile $\vec{m}$ on $X$. Then the following holds. 
	\begin{enumerate}
		\item[(i)] $F(U(N))$ and any resolution of $F(U(N))$ is an attainment of $\vec{m}$. 
		\item[(ii)] $N$ is semi-stable.
		\item[(iii)] Suppose $N$ contains a cut-arc $a$ and $N_1$
		and $N_2$ are the connected components of $N$ obtained
		by deleting $a$. If $\rho_N\in V(N_1)$ and $x\not\in L(N_1)$ then $N_1^x$ is an attainment of $\vec{m}(N_1^x)$
		and $N_2$ is an attainment of $\vec{m}(N_2)$.
		\end{enumerate}
	\end{lemma}
\begin{proof} 
(i): Clearly,
	$U(N)$ is the unfold of $N$ and also of
	$F(U(N))$. \blue{In view of  
	(R2)},
	we obtain $h(F(U(N)))\leq h(N)$. Since $N$ is a attainment of $\vec{m}$
	and $F(U(N))$ realizes $\vec{m}$ it follows that 
	\blue{$h(N)\leq h(F(U(N)))$ must hold too. Thus,}
	$h(F(U(N)))= h(N)$. \blue{Consequently,} $F(U(N))$ is an attainment of $\vec{m}$. To see the remainder, suppose for contradiction that $F(U(N))$ 
	has a resolution $D$ that is not an attainment of $\vec{m}$. Then $h(D)=h(F(U(N)))<h(D)$; a contradiction.
	
	(ii): \blue{Since $N$ is an attainment of $\vec{m}$ it cannot contain
		a pair of identifiable vertices as otherwise $h(F(U(N)))< h(N)$ would hold 
		which is impossible in view of Assertion~(i). By (R3) it follows
		that the compressed networks $C(N)$ and $C(F(U(N)))$ are equivalent.
		Hence $N$ must be a resolution of $F(U(N))$.}
	
 (iii): \blue{Since} $a$ is a cut-arc of $N$ \blue{and therefore cannot have a head
 that is a hybrid vertex,} we have $h(\vec{m})=h(\vec{m}(N_1^x))+ h(\vec{m}(N_2))$. \blue{Since every directed path from the root of $N$ to a leaf of $N_2$ must cross $a$ because $a$ is a cut-arc of $N$ it follows that} $m_N(y)=m_{N_1^x}(x)\times m_{N_2}(y)$
	holds for all $y\in L(N_2)$. \blue{This implies the statement.}
\end{proof}


The \blue{unfold and fold-up operations} described in Section~\ref{section:fun} lie at the heart of the proof of Proposition~\ref{prop:key}.

\begin{proposition}\label{prop:key}
	Suppose $\vec{m}$ is a ploidy profile 
	on $X=\{x_1,\ldots, x_n\}$ and that $N$ is an attainment of $\vec{m}$.
	Then there must exist a directed path $P$
	from the root \blue{of $F(U(N))$ to $x_1$} in $F(U(N))$ such that every hybrid vertex
	in $F(U(N))$ lies on $P$. If, in addition, $N$ 
	is stable then $P$ must be a \blue{directed} path in $N$.
	\end{proposition}
	
	\begin{proof}
		\blue{Put $\vec{m}=(m_1,\ldots, m_n)$.}
			Suppose for contradiction that there exists no  directed path from the root
			\blue{$\rho$ of $F(U(N))$} to $x_1$ in $F(U(N))$ that contains all hybrid vertices of $F(U(N))$. Then since $N$ is an attainment of $\vec{m}$, Lemma~\ref{lem:folding} implies that $F(U(N))$
			is also an attainment of $\vec{m}$. Consequently, $h(N)=h(F(U(N)))$.
			Let $\gamma_{U(N)}:T_1,T_2,\ldots, T_l$, some \blue{$l\geq 1$}, denote a \blue{guide} sequence for \blue{$F(U(N))$}. 
			\blue{Without loss of generality we may assume that $l\geq 2$ since otherwise $F(U(N))$
			only contains one hybrid vertex and, so, the proposition holds.}
			Then there must exist some $i\in\{2,\ldots, l\}$ such that $T_i$ is not a subMUL-tree of $T_{i-1}$ as otherwise all hybrid vertices of $F(U(N))$ would lie on a directed path from $\rho$ to $x_1$. Without loss of generality, we may assume that
			$i$ is as small as possible with this property, 
			i.\,e.\, $T_{j+1}$ is a subMUL-tree of $T_j$,
			for all $1\leq j\leq i-2$.
			
			Let $M$ denote the MUL-tree obtained from $U(N)$ as follows.
			For $j\in\{1,i\}$ let  $t_j$ denote the number of equivalent copies of $T_j$ in $U(N)$. Let $t=\min\{t_1,t_i\}$. \blue{ Then $t\geq 2$.}
			\blue{Choose} $t$ equivalent copies $R_1,\ldots, R_t$
			 of $T_i$ in $U(N)$. For all $1\leq j\leq t$,
			 delete the incoming arc of the root $r_j$ of $R_j$.
			Next choose $t$ equivalent copies of $T_1$ in $U(N)$
			and, for  all  $1\leq j\leq t$, subdivide the incoming arc of the root of $T_j$
			by a vertex $s_j$. Note that this is possible since $T_1$ is the first element in $\gamma_{U(N)}$ 	and so cannot be $U(N)$. 
			\blue{Last-but-not-least,} add the arcs $(s_j,r_j)$, for all $1\leq j\leq t$.
		\blue{Since this might have resulted in arcs whose head is not contained 
			in $X$ and also vertices that have indegree one and outdegree one,
			we clean the resulting MUL-tree by removing the former and repeatedly suppressing the latter. Also we repeatedly identify the root with its unique
			child if this has rendered it a vertex with outdegree one. }
			
			By construction, $F(M)$ is a phylogenetic network 
			that realizes $\vec{m}$. Furthermore,
			$h(F(M))=h(F(U(N)))-(t-1)=h(N)-(t-1)<h(N)$ must hold since $t\geq 2$; a contradiction as $N$ is an attainment of $\vec{m}$.
			
			\blue{The remainder of the proposition is an immediate consequence because $N$ and $F(U(N))$ are equivalent in this case.}
\end{proof}

Since, as mentioned above, beaded trees are
	stable phylogenetic networks the corresponding result for beaded trees in \cite[Lemma 13]{IJJMZ19} is
	a consequence of Proposition~\ref{prop:key} \blue{(once an incoming arc has
		been added to the root)}.

\begin{lemma}
	Suppose $\vec{m}=(m_1,\ldots, m_n)$ is a simple ploidy profile on $X$
	such that $m_1$ is a prime number. Then any cut-arc in an
	attainment of $\vec{m}$
	must be trivial.
	\end{lemma}
	\begin{proof}	
		Suppose $N$ is an attainment of $\vec{m}$. Then 
	\green{the phylogenetic network $N'$ obtained from $N$ by removing, 
for all $2\leq i\leq n$, the cut arcs ending in a leaf $x_i$ of $N$ as well as the leaves $x_i$ (suppressing the resulting vertices of indegree one and outdegree one \orange{and also the root in case this has rendered it an outdegree one vertex}) is a phylogenetic network on $X'=\{x_1\}$. Note that since 
		none of the elements $x_i$
		indexing $m_i$, $2\leq i\leq n$, contributes to $h(N)$,
		we have $h(N)=h(N')$. Thus, $N'$ is an attainment of
		the ploidy profile $\vec{m_1}=(m_1)$.}
		 Put $m=m_1$ and $x=x_1$. If $m\in \{2,3\}$ then the lemma clearly holds since the only cut arc of $N'$ is the incoming arc of $x_1$  and therefore is trivial. So assume that $m\geq 4$.
		 
		 Assume for contradiction that $N'$ has a 
		 non-trivial cut-arc $a$.
		 Let $N_1$ and $N_2$ denote the connected components \blue{of $N'$} obtained by deleting $a$. Assume 
		 	without loss of generality that \blue{the root of $N'$} is contained in  $V(N_1)$. Let $y\not\in L(N_1)$.
		 	Then since for all leaves $z$ in a phylogenetic
		 	network $M$ the number of directed paths from the root of $M$ to $z$ is $m_M(z)$ it follows that  $m=m_{N'}(x)=m_{N_1^y}(y)\times m_{N_2}(x)$. Since
		 $1\not\in\{m_{N_1^y}(y), m_{N_2}(x)\}$ and $m$ is prime this is
		 impossible.
		\end{proof}

\section{\blue{Realizing} simple ploidy profiles}
\label{sec:simple}

\blue{We start this section with associating to a simple ploidy profile} $\vec{m}$  a \blue{binary phylogenetic network $D(\vec{m})$}
	that is based 
on the prime factor decomposition of $m_1$ and also a 
\blue{binary phylogenetic network $B(\vec{m})$}
 that is based on the unique bitwise representation of $m_1$. As we shall see,
\blue{other ways to define binary realizations of $\vec{m}$ that are based on the
	prime factor decomposition of $m_1$ or on the bitwise representation of $m_1$
	and that are similar in spirit to the definitions of $D(\vec{m})$ and $B(\vec{m})$ are conceivable.
	Furthermore,  the ploidy profiles considered in Figure~\ref{table4.2}
	suggest that
the relationship between the number of hybrid vertices in $D(\vec{m})$
and in $B(\vec{m})$ is not straight forward.}

\blue{Suppose that $\vec{m}=(m_1,\ldots, m_n)$, $n\geq 1$, is a ploidy profiles
on $X=\{x_1,\ldots, x_n\}$.}
\begin{figure}[ht] 
	\centering 
	\includegraphics[scale=0.6]{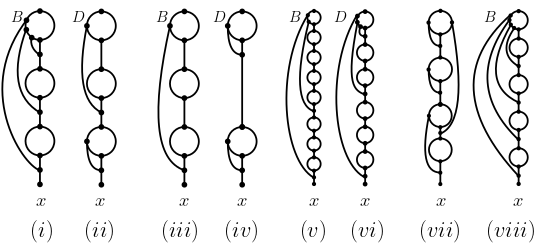}
	
	\caption{For a strictly simple ploidy profile $\vec{m}$ we depict in (i), (iii), (v) and {\blue(viii)} \blue{the phylogenetic network $B=B(\vec{m})$} and in (ii), (iv), and (vi) \blue{the phylogenetic network $D=D(\vec{m})$}. (i) and (ii): $\vec{m}=(15)$ and $h(B)=6>5=h(D)$; (iii) and (iv):
		$\vec{m}=(9)$ and $h(B)=4=h(D)$;
		(v) and (vi):
		$\vec{m}=(265)$ and $h(B)=10<11=h(D)$. \blue{(vii) A realization of the ploidy profile $\vec{m}=(47)$ that uses eight hybrid vertices. (viii) The realization of the ploidy profile 
		in (vii) in terms of $B(\vec{m})$.}
	\label{table4.2}
}
\end{figure}


	

\subsection{The phylogenetic network $D(\vec{m})$}
We begin with introducing further
terminology. Suppose that $m$ is a positive integer and 
that, for all $1\leq i\leq k$, $p_i$ is a prime and
$\alpha_i\geq 1$ is an integer such that 
$p_1^{\alpha_1}p_2^{\alpha_2}\cdot \ldots \cdot p_k^{\alpha_k}$ is 
a prime factor decomposition  of $m$. Without loss of generality, we may 
assume throughout the \blue{remainder} of the paper that 
	the primes $p_1,\ldots, p_k$ are indexed in such a way 
	that \blue{$p_i>p_{i+1}$}
	holds for all $1\leq i\leq k-1$.

For all $1\leq i\leq k$, let $\vec{p}_i=(p_i)$ denote the strictly simple ploidy profile on 	$Y=\{x_1\}$. Also let 
	$\mathcal A(\vec{p}_i)$ denote 
	a binary phylogenetic network on $Y$ that attains
	\blue{$\vec{p}_i$}. Note that $\mathcal A(\vec{p}_i)$ need not be unique. 
		For all $1\leq i\leq k$, we then define a \blue{binary} phylogenetic network $\mathcal A(\vec{p}_i)^{\alpha_i}$ on $Y$ as follows:
	
\subsubsection{\blue{The phylogenetic network $\mathcal A(\vec{p}_i)^{\alpha_i}$} }
	 We take the root $\rho_i$ of $\mathcal A(\vec{p}_i)$ to be the root of 
	$\mathcal A(\vec{p}_i)^{\alpha_i}$.
If $\alpha_i=1$ then we take 
$\green{\mathcal A(\vec{p}_i)^{\alpha_i}}$ to be $\mathcal A(\vec{p}_i)$.
If $\alpha_i\geq 2$ then we make $\alpha_i$ equivalent
copies of $\mathcal A(\vec{p}_i)$ and order them  in some way. Next, we 
identify the unique leaf of the first of the $\alpha_i$ copies of $\mathcal A(\vec{p}_i)$ under that ordering
with the root of the second copy of $\mathcal A(\vec{p}_i)$ and so on until we have processed 
all $\alpha_i$ copies of $\mathcal A(\vec{p}_i)$ this way. The resulting directed acyclic graph is
	$\mathcal A(\vec{p}_i)^{\alpha_i}$ in this case.

To illustrate this construction, assume that $m=4$. 
	Then $k=1$, $p_1=2=\alpha_1$, and $Y=\{x_1\}$. Furthermore, \green{the phylogenetic network depicted in \blue{Figure~\ref{fig:fun-construct}(iv)}
	with the leaf $x_2$ and its incoming arc removed, and the resulting vertex of indegree and outdegree one suppressed, is $\mathcal A(\vec{p}_1)^{\alpha_1}$. }

\subsubsection{\blue{From $\mathcal A(\vec{p}_i)^{\alpha_i}$ to  $D(\vec{m})$ in case $\vec{m}$ is strictly simple}}
\blue{Suppose $\vec{m}$ is strictly simple. Then we obtain $D(\vec{m})$ by} `stacking' the networks $\mathcal A(\vec{p}_1)^{\alpha_1},\ldots, \mathcal A(\vec{p}_k)^{\alpha_k}$
	obtained as described above for a prime factor decomposition $p_1^{\alpha_1}p_2^{\alpha_2}\cdot\ldots\cdot p_k^{\alpha_k}$
	of $m=m_1$ and a choice of attainment $\mathcal A(\vec{p}_i)$ of $\vec{p}_i=(p_i)$, for all $1\leq i\leq k$. 
If $k=1$ then $D(\vec{m})$ is $\mathcal A(\vec{p}_1)^{\alpha_1}$. So assume $k\geq 2$. Then we define \blue{$D(\vec{m})$ to be the phylogenetic
	 network on $\{x_1\}$} obtained by identifying, for all
$1\leq i\leq k-1$, the 
unique leaf of \green{$\mathcal A(\vec{p}_i)^{\alpha_i}$} with the root of
$\mathcal A(\vec{p}_{i+1})^{\alpha_{i+1}}$. 

\blue{For the convenience of the
reader,  we depict $D(\vec{m})$ for the strictly simple ploidy profile
	$\vec{m}=(9)$  on $\{x\}$ in  Figure~\ref{table4.2}(iv).} 

\subsubsection{\blue{From $\mathcal A(\vec{p}_i)^{\alpha_i}$ to $D(\vec{m})$ in case $\vec{m}$ is not strictly simple}}
For all primes $p$ in the prime factor
	decomposition of $m_1$, choose a \blue{binary} attainment $\mathcal A(\vec{p})$ 
	\blue{of the strictly simple ploidy profile $\vec{p}=(p)$}
	and  construct the network $D(\vec{m'})$ for the strictly simple ploidy profile $\vec{m'}=(m_1)$ \blue{as described above}. That network we then process further as follows. First, \blue{we choose an outgoing arc $a$
		of the root of $D(\vec{m'})$ and subdivide \blue{it} with $n-1$ subdivision
		vertices $s_2, \ldots, s_n$ where, starting at the tail of $a$, the first
		subdivision vertex is $s_2$, the next is $s_3$, and so on. }
	To the vertices $s_i$, $2\leq i\leq n$  we then add
the arcs $(s_i,x_i)$  to obtain \blue{$D(\vec{m}$)}.


\blue{As an immediate consequence of the} construction of \blue{$D(\vec{m}) $,
	we have  that $D(\vec{m})$ does not contain an identifiable pair
	of vertices. In view of (R1) 
	it follows that $D(\vec{m})$ is semi-stable. In summary, we therefore have
	the following result.}

\begin{lemma}\label{lem:key}
	Suppose \blue{$\vec{m}$} is a simple ploidy profile on $X$. 
	\blue{Then $D(\vec{m})$} is a binary, semi-stable phylogenetic network on $X$ that realizes $\vec{m}$.
\end{lemma}

\blue{Note that as the strictly simple ploidy profile 
	$\vec{m}=(m)$ with $m=265$ shows, the phylogenetic network
depicted  in Figure~\ref{table4.2}(v) uses fewer hybrid vertices to attain $\vec{m}$
than the phylogenetic network $D(\vec{m})$ depicted in Figure~\ref{table4.2}(vi). Thus, an attainment of a simple ploidy profile $\vec{m}$ need not be obtained from a prime factor decomposition of 
the first component of $\vec{m}$.}

\blue{For the remainder of this section, assume again that $\vec{m}=(m_1,\ldots, m_n)$,
$n\geq 1$ is a simple ploidy profile on $X=\{x_1,\ldots, x_n\}$.}

\subsection{The phylogenetic network $B(\vec{m})$}
\label{sec:B(M)-def}

We start with associating two vectors to a positive
	integer $m$ which we call  the {\em bitwise representation (of $m$)} and the {\em binary representation (of $m$)}, \blue{respectively}.	For $m$ a positive integer, the  first is the 0-1 vector	$\vec{v}_m= (v_m^f,\ldots, v_m^1, v_m^0)$ such that $m=\sum_{i=0}^f 2^i v_m^i$.
	For ease of presentation, and unless stated
	otherwise, we denote by $v_m^f$ the most significant bit that is one. The second is
the vector $(i_1,\ldots, i_q)$, $q\geq 1$ and $i_j\not=0$, for all $1\leq j\leq q-1$,
such that  $m=\sum_{j=1}^q 2^{i_j}$ holds. Informally speaking, the $j$-th entry of that vector is the exponent
of the term $2^{i_j}$ in the bitwise representation of $m$. \blue{Note that $2^{i_1}$ indexes the component $v_m^f$ of  $\vec{v}_m $.}
For example, the bitwise representation of $m=11$ is $(1,0,1,1)$ and the binary representation of $m$ is $(3,1,0)$.


\subsubsection{ \blue{The phylogenetic network $B(\vec{m})$ in case $\vec{m}$ is strictly simple}}
\blue{ Then} $\vec{m}=(m_1)$ and $X=\{x_1\}$. \blue{Let $B(q)$}
	denote \blue{the} beaded tree with unique leaf \blue{$x_1$ and} 
	$q\geq 0$ hybrid vertices.
Let $(i_1,\ldots,i_q)$ denote the binary representation
of $m_1$. \blue{Then $B(\vec{m})$} is obtained from \blue{the beaded tree} $B(i_1)$
as follows. \blue{Choose one the two outgoing arcs of the
	root of $B(i_1)$ and subdivide it with $q-1$ vertices}
$s_2,\ldots, s_q$ not contained in $B(i_1)$ \blue{so that $s_2$ is the child of the
	root of $B(i_1)$, $s_3$ is the child of $s_2$, and so on.}
For all $1\leq j\leq q$, we then add an arc
$a_j$ to $s_j$ whose head is a subdivision vertex  of
the outgoing arc of the hybrid vertex of $B(i_1)$ that has precisely 
$i_j$ hybridization vertices of $B(i_1)$ strictly below it.  

We refer the interested reader to Figure~\ref{table4.2}\blue{(iii)} for an illustration of \blue{$ B(\vec{m})$ for the strictly simple ploidy profile $\vec{m}=(9)$}.

\subsubsection{ \blue{The phylogenetic network $B(\vec{m})$ in case $\vec{m}$ is not strictly simple}}
We first construct \blue{the} phylogenetic network $B(\vec{m'})$ for the strictly simple ploidy profile $\vec{m'} =(m_1)$ on $\{x_1\}$. \blue{Next, we choose one of the two
outgoing arcs \orange{of the root} of $B(\vec{m'})$ and subdivide that arc with $n-1$ 
subdivision vertices $t_2,\ldots, t_n$ such that $t_2$ is the child of the
root of $B(\vec{m'})$, $t_3$ is the child of $t_2$ and so on.}
 Finally, we attach to each $t_i$ the
arc $(t_i,x_i)$, $2\leq i\leq n$. 

To illustrate this construction,
consider the simple ploidy profile $\vec{m}_1=(5,1)$
on  \blue{$X'=\{x_1,x_2\}$. Then 
$\vec{m'}=(5)$ and the phylogenetic network $D$ depicted in 
Figure~\ref{fig:new} is 
$B(\vec{m})$. In fact, $B(\vec{m})$ is a \blue{binary} attainment of $\vec{m}$.} 

As indicated in Figure~\ref{table4.2}, the relationship between \blue{$D(\vec{m})$,
 $ B(\vec{m})$}, and a \blue{binary} attainment of a simple ploidy profile $\vec{m}$ is far from clear in general. This holds even if
$\vec{m}=(m)$ is strictly simple and $m$ is a prime. Indeed
for $m=47$ the hybrid number of $\vec{m}$ is at most eight
\blue{since the phylogenetic network depicted in Figure~\ref{table4.2}(vi)
	realizes $\vec{m}$.}
However  \blue{$h(B(\vec{m}))=9$}. This implies that, in general, 
\blue{$ B(\vec{m})$ with $\vec{m}=(p)$} and $p$ a prime cannot be
used as an attainment with which to initialize the construction of \blue{$D(\vec{m})$.}

As an immediate consequence of the construction of \blue{$B(\vec{m})$,
we have the following companion result of Lemma~\ref{lem:key} since 
	similar arguments as in the case of $D(\vec{m})$ imply that
	$B(\vec{m})$ is semi-stable}. 

\begin{lemma}\label{lem:length-2}
	Suppose \blue{$\vec{m}$} is a simple ploidy profile on $X$. \blue{Then $B(\vec{m})$}
	is a binary, semi-stable phylogenetic network on $X$ that
	 realizes $\vec{m}$.
\end{lemma}

To gain insight into the structure of $B(\vec{m})$, we next
present formulae for counting, for a simple ploidy profile $\vec{m}$,
the number $b(\vec{m})$ of vertices in $B(\vec{m})$ and also the number of hybrid vertices of $B(\vec{m})$. Note that such formulae are known for certain types of  phylogenetic networks without beads (see e.g.\cite{MSW15,IK11} and
\cite{S16} for more).  To state them, we require further terminology.
Suppose $m\geq 1$ is an integer and 
$\vec{v}_m$ is the bitwise representation of $m$.
Then we 
denote by $p(m)$ the number of non-zero bits in
$\vec{v}_m$ bar the first one. For example,
if $m=6$ then $p(m)=1$. \blue{Furthermore, we denote the dimension of
a vector $\vec{v}$ by $\dim(\vec{v})$.}
%

%
Armed with this, the construction of \blue{$B(\vec{m})$} from 
	a simple ploidy profile $\vec{m}$ implies \blue{our first main result.}

\begin{theorem}\label{Bm-vertex-count}
	Suppose that $\vec{m}=(m_1, m_2,\ldots, m_n)$, \blue{$n\geq 1$,} is a simple ploidy profile.
	\blue{Let $\vec{i}_{m_1}=(i_1,i_2,\ldots, i_l)$, some $l\geq 1$, denote the binary representation of $m_1$. } 
		Then 
		$$\orange{b(\vec{m})}=2(\blue{i_1+\dim(\vec{i}_{m_1})-1} + n-1) + 1\blue{=2(\dim(\vec{v}_{m_1})-1+p(m_1)+n-1)+1}
		$$
		Furthermore, $B(\vec{m})$ has  $\blue{i_1+\dim(\vec{i}_{m_1})-1}$ hybrid  vertices.
\end{theorem}

\blue{We remark in passing that in case $\vec{m}=(m)$ is strictly simple then any binary phylogenetic network $N$ that realizes $\vec{m}$
has  $2h(N)+1$ vertices since $N$ has only one leaf and, so, the number of tree vertices of $N$ plus the root
	must equal its number of hybrid vertices. Note that in case $N$ is $B(\vec{m})$ then this also follows
	from Theorem~\ref{Bm-vertex-count} since $n=1$ and $ i_1+\dim(\vec{i}_{m_1})-1$ is the number of hybrid vertices
	of $N$ and therefore also the number of tree vertices of $N$ plus the root.}

\section{\blue{Realizing general ploidy profiles}}
 \label{sec:n(m)}
 
To help establish 
a formula for computing the hybrid number of a ploidy profile, we start by associating a binary phylogenetic network $N(\vec{m})$ 
on $X$ to a ploidy profile $\vec{m}$ on
$X$ that realizes $\vec{m}$. This network is
recursively \green{obtained} via a two-phase process which we present in the form of pseudo-code in
	Algorithms~\ref{alg:simplification-seq} (Phase I) and \ref{alg:Nmalgorithm} (Phase II).
	We next outline both phases and refer the reader  \blue{to Figure~\ref{fig:modification} 
	for an illustration of the three cases considered in Algorithm~\ref{alg:Nmalgorithm} 
and}	to Figure~\ref{fig:new} for an illustration of the construction of $N(\vec{m})$ from the ploidy profile $\vec{m}=(12,6,6,5)$. The \blue{phylogenetic network} $D$ in that figure is the phylogenetic network with which the construction of $N(\vec{m})$ is initialized.

Suppose $\vec{m}=(m_1,\ldots m_n)$ is a ploidy profile on $X$. \green{Then, in Phase I,} we \blue{iteratively} generate a simple ploidy profile 
$\vec{m}_t$ from $\vec{m}$. 
This process is captured via
a sequence $\sigma(\vec{m})$ of ploidy profiles which we call \blue{the} {\em simplification sequence}
for $\vec{m}$ \blue{and formally define as the output of Algorithm~\ref{alg:simplification-seq} when given $\vec{m}$ as input}. The first element of $\sigma(\vec{m})$ is $\vec{m}$ and the last element is 
\blue{a simple ploidy profile which we call the {\em terminal element} of $\sigma(\vec{m})$ and denote
by} $\vec{m}_t$. 
We denote the number of elements of $\sigma(\vec{m})$ other than $\vec{m}$ by $s(\vec{m})$. Note that if $\vec{m}$ is a simple ploidy profile then $s(\vec{m})=0$ as $\vec{m}=\vec{m}_t$ holds in this case.
 Informally speaking, the purpose of 
	\blue{$\sigma(\vec{m}): \vec{m}_0=\vec{m},\vec{m}_i,\ldots \vec{m}_{s(\vec{m})}=\vec{m}_t$} is to allow us to construct, for all $0\leq i\leq s(\vec{m})$,
	the network $N(\vec{m}_i)$ from $N(\vec{m}_{i+1})$
	by reusing $N(\vec{m}_{i+1})$ (or parts of it) as much as possible
	\blue{(see \cite{HM22} for more on such sequences)}.
	
To formally state Algorithm~\ref{alg:simplification-seq}, we require further \blue{notations. Suppose $\vec{m}=(m_1,\ldots, m_n)$ is a ploidy profile on $X$.
	Then we denote for all $1\leq i\leq n$ the element of $X$ that indexes $m_i$ by  $x(m_i)$. Furthermore, for any non-empty sequence
	 $\sigma$ and any $z$, we denote by $\sigma\cup\{z\}$ the sequence obtained  by adding $z$ to the end of $\sigma$.}

	\begin{algorithm}[h!]
		\caption{\label{alg:simplification-seq} \blue{The} simplification sequence \blue{of}  a ploidy profile.}
		\begin{algorithmic}[1]
			\Require{A ploidy profile $\vec{m} = (m_1, m_2, \ldots, m_n)$ on $X = \{x_1,x_2,\ldots,x_n\}$,  \blue{$n \geq 1$.}
			}
			\Ensure{The simplification sequence $\sigma(\vec{m})$
			for $\vec{m}$ and \blue{a} set $X(\vec{m})$ 
				that contains, \blue{for all ploidy profiles $\vec{m}'$ in $\sigma(\vec{m})$, 
				the set $X'$ that indexes $\vec{m}'$.}
			}
			\State Put $\vec{m}_0\leftarrow \vec{m}$, $\sigma(\vec{m_0})\leftarrow \vec{m_0}$, \blue{$X_0\leftarrow X$, $X(\vec{m_0})\leftarrow\{X_0\}$, and  $k\leftarrow n$.}
			\blue{
			\If{$\vec{m}_0$ is simple} 
			\State Return $\sigma(\vec{m}_0)$ and $X(\vec{m}_0)$.
			\EndIf
		}
			\blue{\While{$\vec{m}=(m_1,\ldots, m_k)$ is not simple}
				\label{alg:loop}
				\State  Put $\alpha=m_1-m_2$ and compute a ploidy profile $\vec{m'}$ on a set $X'$ as follows:
				\If{$\alpha = 0$}
				\State $\vec{m'}=(m_2,m_3,\ldots, m_k) $ and $X'=\{x(m_2), x(m_3),\ldots, x(m_k)\}$.
				\EndIf
				\If{$\alpha >m_2$} \label{simp:alpha-larger}
				\State $\vec{m'}=(\alpha, m_2,m_3,\ldots, m_k) $ and $X'=\{x(\alpha), x(m_2), x(m_3),\ldots, x(m_k)\}$.
								\EndIf
			\If{$\alpha \leq m_2$}
		  \If{there exists some $j\in\{1,\ldots, k-1\}$ so that $m_{j+1} < \alpha \leq m_j$}
		  \State $\vec{m'}=(m_2, m_3, \ldots,  m_j, \alpha, m_{j+1},\ldots, m_k)$ and
		  $X'=\{  x(m_2), x(m_3), \ldots$,  $x(m_j),  x(\alpha), 
		  x( m_{j+1}), \ldots, x(m_k)    \}$.
		  \EndIf
		  \If{ $\alpha = m_k$}
		  \State $\vec{m'}=(m_2, m_3, \ldots,  m_k, \alpha)$ and
		  $X'=\{  x(m_2), x(m_3), \ldots, x(m_k), x(\alpha)\}$.
								\EndIf
								\EndIf
			\State Put $\sigma(\vec{m_0})\leftarrow\sigma(\vec{m_0})\cup\{\vec{m'} \}$, $X(\vec{m_0})\leftarrow X(\vec{m_0})\cup \{X'\}$, $k\leftarrow |X'|$, and	$\vec{m}\leftarrow\vec{m'}$
			and return to Line~\ref{alg:loop}. 
			\EndWhile
			\State 
			Return 
			$\sigma(\vec{m}_0)$ and $X(\vec{m}_0)$.
		}
			\label{alg:alg1-last-line}
		\end{algorithmic}
	\end{algorithm}
	
Phase~II is concerned with generating
 the \blue{phylogenetic} network $N(\vec{m})$ from 
	the simplification sequence of $\vec{m}$ and the set \blue{$X(\vec{m})$} (for both see Phase~I), and an attainment $\mathcal A(\vec{m}_t)$ of $\vec{m}_t$. 
	Note that in case an attainment for $\vec{m}_t$ is not known, we can always initialize the construction
	of $N(\vec{m})$ with \blue{$D(\vec{m})$ or $B(\vec{m})$.
	 The} number of hybrid vertices of the generated network in this case is an upper bound on $h(N(\vec{m}))$ \blue{and therefore also on the hybrid number of 
 $\vec{m}$}.
	   
	To obtain $N(\vec{m})$, we use a trace-back through $\sigma(\vec{m})$
	starting with $\vec{m}_t$.
	More precisely, assume that \blue{$\vec{m}_i=(m_1,\ldots, m_k)$, some $k\geq 2$
	and $\vec{m}_{i+1}$ are two ploidy profiles in $\sigma(\vec{m})$, some $0\leq i\leq s(\vec{m})-1$.} Then to obtain $N(\vec{m}_i)$ from $N(\vec{m}_{i+1})$
we distinguish again
between the cases that \blue{$\alpha:=m_1-m_2=0$, $\alpha >m_2$ and
	$\alpha\leq m_2$}\orange{,} see Figure~\ref{fig:modification}. Note that 
\blue{there} might \blue{be} non-equivalent attainments of $\vec{m}_t$
with which to initialize the construction of $N(\vec{m})$. 
	
\begin{algorithm}[h!]
	\caption{\label{alg:Nmalgorithm} \blue{The} construction of \blue{the} phylogenetic network $N(\vec{m})$ from a 
	ploidy profie $\vec{m}$ \blue{and an attainment for $\vec{m}_t$.}
}
	\begin{algorithmic}[1]
	\Require{A ploidy profile $\vec{m}$ on $X$, an
			 attainment $\mathcal A(\vec{m}_t)$ of  $\vec{m}_t$,
		 and the output of Algorithm~\ref{alg:simplification-seq}
	 }		 
		\Ensure{\blue{The} phylogenetic network $N(\vec{m})$ 
			constructed from $\mathcal A(\vec{m}_t)$. }
		\State \blue{Put $\vec{m}_0\leftarrow \vec{m}$, $\vec{m'}\leftarrow \vec{m}_t$, and $N(\vec{m'})\leftarrow\mathcal A(\vec{m}_t)$.}
\blue{	\If{$\vec{m'}=\vec{m}_0$} 
	\State return $N(\vec{m'})$.
	\EndIf
}
	\blue{
		\While{ $\vec{m'}\not=\vec{m}_0$ } \label{line:loop}
	\State let $\vec{m}=(m_1,\ldots, m_l)$ denote the predecessor of $\vec{m'}=(m_1',\ldots, m_k')$ in $\sigma(\vec{m_0})$,
	some $k$ and some $l$.
	Put $\alpha=m_1-m_2$ and construct the phylogenetic network $N(\vec{m})$ from 
	$N(\vec{m'})$ as follows.
	\If{$\alpha=0$} \label{alpha=0}
	\State for all $2\leq i\leq k$, relabel the leaf $x(m_i')$ of $N(\vec{m'})$ by $x(m_{i+1})$. Replace the leaf $x(m_1')$
	of $N(\vec{m'})$ by the cherry $\{x(m_1),x(m_2)\}$. \label{cond-alpha=0}
	\EndIf
	\If{$\alpha>m_2$} \label{alpha-larger}
	\State
	 for all $1\leq i\leq k$, relabel the leaf $x(m_i')$ of $N(\vec{m'})$ by $x(m_i)$. Subdivide the incoming arcs of  leaves $x(m_1)$
	 and $x(m_2)$ by  vertices $u$ and $v$, respectively, and add the arc $(v,u)$.
\EndIf
\If{$\alpha\leq m_2$} \label{alpha-smaller}
		\State 
		let $j$ be such that $m_{j+1}<\alpha\leq m_j$.
		Subdivide the incoming arc of $x(m_j')$ by a new vertex $v$ and replace 
		$x(m_1')$ by the cherry $\{x(m_1), x(m_2)\}$.
		Subdivide the incoming arc of $x(m_1)$ by a new vertex $u$.
		Add the arc $(v,u)$ and delete $x(m_j')$ as well as its incoming arc $(v,x(m_j'))$
		(suppressing $v$ as $indeg(v)=1=outdeg(v)$ now holds).
		For all $2\leq k\leq j-1$, put $x(m_{k+1})\leftarrow x(m_k')$ and, for all remaining $k$, put $x(m_k)\leftarrow x(m_k')$. \label{cond-alpha-smaller}
		\EndIf
		\State Put $\vec{m'}\leftarrow \vec{m}$ and return to line~\ref{line:loop}. 		
		\EndWhile	
	}
\State Return $N(\vec{m})$.
				\label{alg:alg2-last-line} 
\end{algorithmic}
\end{algorithm}
	
\green{	
	\begin{figure}[h!]
		\centering
		\includegraphics[scale=0.23]{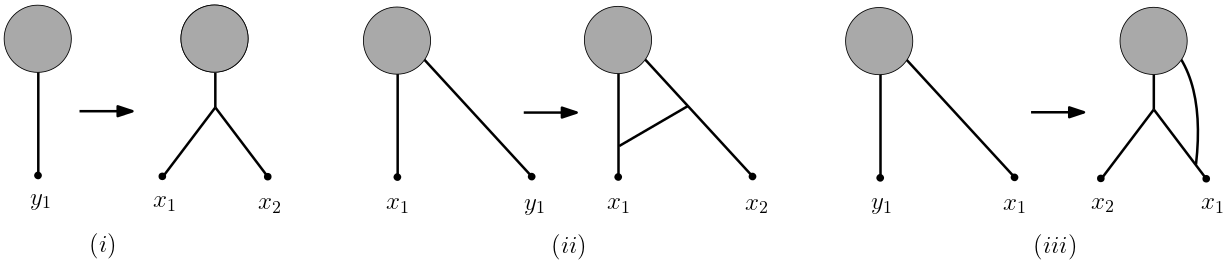}
		\caption{The three cases in the construction of the network
			$N(\vec{m})$ from a ploidy profile $\vec{m}=(m_1,m_2\ldots,m_n)$ considered in
			Algorithm~\ref{alg:Nmalgorithm}. For \blue{$\alpha=m_1-m_2$}, the case \blue{$\alpha=0$
		is depicted in (i),} the 
			case \blue{$\alpha>m_2$ in (ii)}, and  the case 
			\blue{$\alpha\leq m_2$ in (iii)}. \blue{In (iii),} the dashed arc 
			and the vertex \blue{$x(m_j')$} are deleted and the vertex $v$ is suppressed. In each case, the grey disk indicates
			the part of the \blue{phylogenetic}
			network of no relevance to the discussion. 
			\label{fig:modification}}
	\end{figure}
}


To illustrate the construction of $N(\vec{m})$,
\green{consider the 
ploidy profile $\vec{m}=(12,6,6,5)$ on 
$X=\{x_1,\ldots,x_4\}$. Then $\vec{m}$,
$(6,6,6,5)$, $(6,6,5)$, $(6,5)$, $(5,1)$
is the simplification sequence $\sigma(\vec{m})$
associated to $\vec{m}$ \orange{because, by definition, the first element of $\sigma(\vec{m})$ is always 
$\vec{m}$.} The ploidy profile
$(5,1)$ is $\vec{m}_t$. The phylogenetic network $D$ on \blue{$X=\{x_1,x_2\}$}
\blue{on the left of Figure~\ref{fig:new}
is an attainment of $\vec{m}_t$ in the
form of  $B(\vec{m}_t)$.}  Initializing   Algorithm~\ref{alg:Nmalgorithm} with $B(\vec{m}_t)$
yields the phylogenetic network $N(\vec{m})$ at the right of that figure. Apart from the second arrow which
is labelled $(6,5)\to (6,6,6,5)$ as it
combines the steps $(6,5)\to (6,6,5)$ and $(6,6,5)\to (6,6,6,5)$, each arrow is labelled with the 
corresponding traceback step in $\sigma(\vec{m})$. 
}

\begin{figure}[h!]
	\centering
	\includegraphics[scale=0.34]{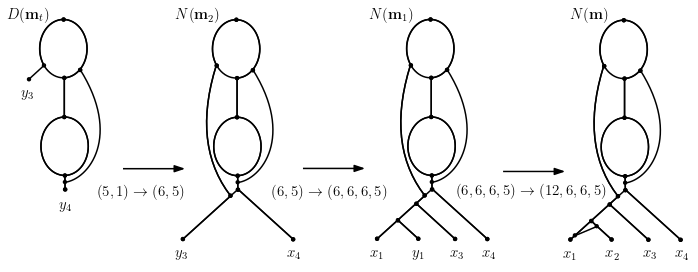}
	\caption{The construction of $N(\vec{m})$ for the ploidy profile
			$\vec{m}=(12,6,6,5)$  \blue{on $X=\{x_1,x_2, x_3, x_4\}$} where we have combined the
			steps $(6,5) \rightarrow (6,6,5)$ and $(6,6,5)\rightarrow (6,6,6,5)$ into the step 
			$(6,5)\rightarrow (6,6,6,5)$. 
			The leftmost network $D$  \blue{on $X'=\{x_1,x_2\}$} is an attainment of  $\vec{m}_t=(5,1)$ 
			\blue{in the form of $B(\vec{m})$} and initializes the construction of $N(\vec{m})$. The
			network $N(\vec{m}_2)$  on $X'$ realizes the ploidy profile $\vec{m}_2=(6,5)$
			and the network $N(\vec{m}_1)$  \blue{on $X$} realizes the
			ploidy profile $\vec{m}_1=(6,6,6,5)$.
			The rightmost network is $N(\vec{m})$. 
			The arrow labels indicate how a ploidy profile
		in $\sigma(\vec{m})$ was obtained. 
						\label{fig:new}}
\end{figure}


For any \blue{attainment} $\mathcal A(\vec{m}_t)$
	of the terminal element $\vec{m}_t$ of the simplification sequence $\sigma(\vec{m})$ of a ploidy profile $\vec{m}$ on $X$, the
	graph $N(\vec{m})$ is a phylogenetic network on $X$ that realizes $\vec{m}$. Also,
	at each step in the traceback through  $\sigma(\vec{m})$ 
	the number of vertices is increased by exactly two.  \blue{Denoting}  the number of vertices
of $N(\vec{m})$ by $n(\vec{m})$ and the number of vertices
in a \blue{binary} attainment $\mathcal A(\vec{m}_t)$ of $\vec{m}_t$ by $a(\vec{m}_t)$, \blue{we obtain our next result}.

\begin{lemma}
	\label{lem:nm-numbers}	
	Suppose $\vec{m}$ is a ploidy profile on $X$. Then for any binary attainment of $\vec{m}_t$ used in the initialization of the construction of $N(\vec{m})$, we  have that $N(\vec{m})$ is a binary phylogenetic network on $X$ that realizes $\vec{m}$. Furthermore,
	$n(\vec{m}) = a(\vec{m}_t) + 2 s(\vec{m})$.
\end{lemma}

\blue{In combination with}
	Theorem~\ref{Bm-vertex-count}, it follows   that $N(\vec{m})$ has at most 
	$b(\vec{m}_t) + 2 s(\vec{m})=2\blue{(i_1+\dim(\vec{i}_{m_1})} + n +s(\vec{m}) +l)-3$ vertices 
	and also at most $\blue{ i_1+\dim(\vec{i}_{m_1})-1}+s(\vec{m})$
		 hybrid vertices where $\vec{m}_t=(m_1,\ldots, m_l)$,
		 some $l\geq 1$, \blue{and $i_1$ is the first component in the binary representation $\vec{i}_{m_1}$ of $m_1$}.
		 %
%
		 Furthermore, we have

	\begin{proposition}\label{prop:upper-bound}
		Suppose $\vec{m}=(m_1,\ldots, m_n)$ is a ploidy profile on $X$ such that 
		\blue{$B(\vec{m}_t)$} is a \blue{binary} attainment of $\vec{m}_t$.
		For all $1\leq k\leq n$, let \blue{$(i_{k,1},\ldots,i_{{k,l_k}})$}
		denote the binary representation  of $m_k$, some 
		$l_k\geq 1$. Then the following holds.
		\begin{enumerate}
			\item[(i)]	$h(\vec{m})\leq \sum_{k=1}^n (\blue{i_{k,1}}+l_k-1)$. In case  $\vec{m}$ is 
			simple, $h(\vec{m})=\blue{i_{1,1}}+l_1-1$ which is sharp. 
			\item[(ii)] If $m_i=2^{\blue{i_{k,1}}} $ holds for all $1\leq k\leq n$ 
			then $h(\vec{m})=\blue{i_{1,1}}$.
		\end{enumerate}
	\end{proposition}
	
	\begin{proof}
		(i)	To see the stated inequality, we construct a 
		binary phylogenetic network $B$ on $X=\{x_1,\ldots, x_n\}$ 
		from $\vec{m}$ as follows. \blue{For all}
		$1\leq k\leq n$, \blue{we first construct  $B_k=B(\vec{m}_k)$ 
		 where $\vec{m}_k$ is the} 
		strictly simple ploidy profile \blue{$(m_k)$}. Next,
		we add a new vertex $\rho$ and, for all $1\leq k\leq n$,		
		an arc from $\rho$ to the
		root of $B_k$. If the resulting phylogenetic network
		on $X$ is binary then that network is $B$. Otherwise, $B$
		is a phylogenetic network obtained by resolving $\rho$
		\blue{so that $\rho$ has outdegree two}. 
		
		\blue{By construction,} $B$ realizes $\vec{m}$ \blue{because
		$B_k$ realizes $\vec{m}_k$, for all $1\leq k\leq n$. By} 
		Theorem~\ref{Bm-vertex-count}, \blue{it follows that}
		$h(B_k)=\blue{i_{k,1}}+l_k-1$. Thus,
		$h(\vec{m})\leq h(B)= \sum_{k=1}^n(\blue{i_{k,1}}+l_k-1)$, as required. 
		If $\vec{m}$ is simple then $k=1$ and so  
		$h(B)=h(B_1)=\blue{i_{1,1}}+l_1-1$.
		
		(ii)
	This is a straight forward consequence of (i) and the
		fact that in this case $B_k$ is the  \blue{beaded tree $B(i_{k,1})$}.
	\end{proof}

Note that as the example of the ploidy profile $(k^l,k)$ for some $l,k\geq 2$ 
	shows, there exists an infinite family of
	ploidy profiles $\vec{m}$ for which the length 
	of the simplification sequence for $\vec{m}$ is at least $k^{l-1}+1$ and therefore grows exponentially in $l$. As a consequence 
	of this, we also have, for any attainment of $\vec{m}_t$, that the number of hybrid vertices
in $N(\vec{m})$ can grow exponentially in $l$.
In view of this, we next study \orange{simplification} sequences
for special types of ploidy profiles. 
To this end we 
call an element $j\in\{1,\ldots, n\}$ {\em maximum} if  \blue{$m_j$ is the
last component of a ploidy profile $\vec{m}= (m_1,\ldots, m_n)$, $n\geq1$, that is not one.}

\begin{proposition}\label{prop-simp-length}
	Suppose $\vec{m}=(m_1,\ldots, m_n)$ is a ploidy profile on $X$. \green{Let $q$ denote the maximum index of $\vec{m}$.} Then the following holds
	\begin{enumerate}
		\item[(i)] If  $k\geq 2$ is an integer such that $m_i=k$ holds for all 
		$1\leq i\leq q$ then $s(\vec{m})=q-1$.
		\item[(ii)] If $k\geq 1$ and  $l\geq q+2$ are integers such that 
		$m_i=k(l-i)$ holds for all $1\leq i\leq q$ then $s(\vec{m})=\blue{l+q-3}$. 
		\end{enumerate}
\end{proposition}
\begin{proof}
	Note first that for both statements, we may assume without loss of
	generality that $q=n$ since elements in $X$ with ploidy number one do not contribute to $s(\vec{m})$.
	
	(i): Since $m_i=m_{i+1}$ holds for all $1\leq i\leq n-1$, the difference 
	in dimension between any two consecutive ploidy profiles in 
	$\sigma(\vec{m})$ is one. Hence, $q-1$ operations are needed to
	transform $\vec{m}$ into $\vec{m}_t$. Consequently, 
	$s(\vec{m})=q-1$.
	
	(ii): Since \green{$m_{i-1}-m_i=k$} holds for all $2\leq i\leq q$,  it follows 
	that $q-1$ operations are needed to transform $\vec{m}$ into a ploidy
	profile \blue{$\vec{m'}$ of the form $(k(l-q), k,\ldots, k,1,\ldots, 1)$
	where the components after the last $k$ may or may not exist. To transform
	$\vec{m'}$ into a ploidy profile $\vec{m''}$ of the from $(k, k,\ldots, k,1,\ldots, 1)$
	a further $l-q-1$ operations are needed. By Assertion (i), a further $q-1$ operations
	are needed to transform $\vec{m''}$ into a simple ploidy profile. Since
	$\sigma(\vec{m}) $ is the concatenation of the underlying simplification sequences
	it follows that $s(\vec{m})=q-1+l-q-1+q-1=q+l-3$.
}
	\end{proof}

Together with Lemma~\ref{lem:nm-numbers}, the next result may be viewed as the companion result of  Lemmas~\ref{lem:key} 
and \ref{lem:length-2} for general ploidy profiles.

\begin{proposition}
	For any ploidy profile $\vec{m}$ on $X$ and any \blue{binary}
	attainment of the terminal element in $\sigma(\vec{m})$, the graph 
	 $N(\vec{m})$ is a binary, semi-stable phylogenetic network on $X$ that realizes $\vec{m}$.
	\end{proposition}
\begin{proof}
	\green{In view of Lemma~\ref{lem:nm-numbers}, it suffices to show that $N(\vec{m})$ is semi-stable.}	
	Assume for contradiction that there exists a ploidy profile
	$\vec{m}=(m_1,\ldots, m_n)$ on $X$ such that 
	$N(\vec{m})$ is not semi-stable. Since
	the construction of $N(\vec{m})$ is initialized with \green{an
	attainment of the terminal element $\vec{m}_t$ of $\sigma(\vec{m}): \vec{m}_0=\vec{m}, \vec{m}_1,\ldots, \vec{m}_l=\vec{m}_t$, some $l\geq 0$ and, by Lemma~\ref{lem:folding}(ii), an attainment is semi-stable} 
	there must exist some $0\leq i\leq l$
	such that the network \green{$N(\vec{m}_i)$ is not semi-stable but all networks
	$N(\vec{m}_j)$, $i+1\leq j\leq l$} are semi-stable. Without
	loss of generality, we may assume that $i=0$. Put 
	$\blue{\vec{m'}=\vec{m}_1}$.
	
	We claim first that $m_1\not=m_2$. Indeed, if $m_1=m_2$ then
	\blue{$\alpha =0$. Hence, Line~\ref{alpha=0} in Algorithm~\ref{alg:Nmalgorithm}
	is executed to obtain $N(\vec{m})$ from $N(\vec{m'})$.}
	Since, by assumption, $N(\vec{m})$ is not semi-stable it follows that
	\blue{$N(\vec{m'})$} is not semi-stable; a contradiction. Thus,
	$m_1\not=m_2$, as claimed.

	We next claim that $m_1>m_2$ cannot hold either. 
	Assume for contradiction that $m_1>m_2$. 
	Put \blue{$\alpha=m_1-m_2$}.   
	Assume first that $\alpha>m_2$. 
	\blue{Then Line~\ref{alpha-larger} in  Algorithm~\ref{alg:Nmalgorithm}
	is executed to obtain $N(\vec{m})$ from $N(\vec{m'})$.}
	Since \blue{$N(\vec{m'})$} is
	semi-stable, and this does not introduce an
	identifiable pair of vertices in $N(\vec{m})$, it 
	follows that $N(\vec{m})$ is also semi-stable
	which is impossible.
 
 So assume that  $\alpha\leq m_2$. \blue{Then Line~\ref{alpha-smaller} in  Algorithm~\ref{alg:Nmalgorithm}
	is executed to obtain $N(\vec{m})$ from $N(\vec{m'})$.}
 %
 \blue{Similar arguments as in the previous two cases imply again a contradiction.}
		This completes the proof of the claim.
	
	Thus, $m_1<m_2$ must hold. Consequently, $\vec{m}$ is not a ploidy profile; a contradiction.  Thus,  $N(\vec{m})$ must be semi-stable.

\end{proof}

\section{The hybrid number of a ploidy profile}
\label{sec:optimal-n(m)}
  \blue{In
	this section, we prove Theorem~\ref{theo:min}
	which implies a closed formula for the hybrid number of a ploidy profile
	(Corollary~\ref{cor:upper-bound}). To help illustrate  our theorem, we remark 
	that for Line~\ref{alpha-larger} in Algorithm~\ref{alg:Nmalgorithm} 
not to be executed we must have for every element $\vec{m'}=(m'_1,\ldots, m'_{n'})$, some $n'\geq 2$,
in the simplification sequence of $\vec{m}$ that $m_1'>2m_2'$ does not hold. }

\begin{theorem}\label{theo:min}
	Suppose $\vec{m}$ is a ploidy profile on $X$ 
	such that, for every ploidy profile \blue{in $\sigma(\vec{m})$, Line~\ref{alpha-larger} in  Algorithm~\ref{alg:Nmalgorithm} is not 
	executed.}
	 If \blue{$\mathcal A(\vec{m}_t)$} is an attainment for $\vec{m}_t$ \blue{with which the construction of $N(\vec{m})$ is
	 initialized}
		then $N(\vec{m})$ is an attainment for $\vec{m}$.
\end{theorem}
\begin{proof}
\blue{Put $\vec{m}=(m_1,\ldots, m_n)$ and assume that $\vec{m}$ is such that $\mathcal A(\vec{m}_t)$ is an
attainment of $\vec{m}_t$. }
	Suppose $X=\{x_1,\ldots, x_n\}$, 
	$1\leq n$.	Note that we may assume 
	that $n\geq 2$ as 
	otherwise $\vec{m}$ is simple. \blue{Hence,
	$\vec{m}=\vec{m}_t$  and, so,}  the theorem follows by assumption 
	on $\vec{m}_t$.  \blue{Similar arguments as before imply that we may also assume that $\vec{m}$ is not simple.}
	
	Assume for contradiction that $N(\vec{m})$ is not an attainment of
		$\vec{m}$. Let $Q$ denote an attainment of $\vec{m}$. Then $h(Q)<h(N(\vec{m}))$.
	 In view of Proposition~\ref{prop:key}, there must exist a directed path $R$ in $F(U(Q))$ from \blue{the root} $\rho$  \blue{of $F(U(Q))$} to $x_1$
	 that contains all hybrid vertices of  \blue{$F(U(Q))$}. Since $h(Q)=h(F(U(Q)))$ as $C(Q)$ and $F(U(Q))$ are equivalent
	 \blue{by (R3)}, it follows that we may also assume that $Q$ is binary and that $R$ gives rise to a
	  path $P$ from $\rho$ to $x_1$ that contains all \blue{hybrid} vertices \blue{of $Q$}.
	
	Since the construction of $N(\vec{m})$ is initialized with an attainment of 
	$\vec{m}_t$, there must exist a ploidy profile \blue{$\overline{\vec{m}}$}
	in $\sigma(\vec{m})$ such that there exists a binary phylogenetic network 
	\blue{$\overline{Q}$} that realizes \blue{$\overline{\vec{m}}$}
	and for which \blue{$h(\overline{Q})<h(N(\overline{\vec{m}}))$} holds.
	Without loss of generality, we may assume that
	\blue{$\overline{\vec{m}}$} is such that for all ploidy profiles
	$\vec{m''}$ succeeding \blue{$\overline{\vec{m}}$} in 
	$\sigma(\vec{m})$ we have $h(N(\vec{m''}))\leq h(Q'')$
	for all binary phylogenetic networks $Q''$ that 
	realize $\vec{m''}$. For ease of presentation we may assume that \blue{$\overline{\vec{m}}=\vec{m}$}.
	
	Put \blue{$\vec{m'}=\vec{m}_1=(m_1',\ldots, m_{l'}')$, some $l'\geq 1$.
	Also, put $\alpha=m_1-m_2$,}
	$N=N(\vec{m})$, and \blue{$N'=N(\vec{m}')$.
	Since Line~\ref{alpha-larger} in Algorithm~\ref{alg:Nmalgorithm} is not executed for any element in
	$\sigma(\vec{m})$, it follows that either $\alpha=0$ or that $\alpha\leq m_2$ since either
	Line~\ref{alpha=0} or Line~\ref{alpha-smaller} of that algorithm must be executed in a 
	pass through the algorithm's while loop.} 
\\
	
\noindent{\bf Case (a):} \blue{Assume that $\alpha=0$. Let $x_1=x(m_1)$ 
	and $x_2=X(m_2)$ as in Line~\ref{cond-alpha=0} in Algorithm~\ref{alg:Nmalgorithm}.} 
 Let $2\leq \blue{r}\leq n$ such that
	$m_1=\blue{m_r}$ holds. By the minimality of
	$h(Q)$ it follows that the induced subgraph $T$ 
	of $Q$ connecting the elements in $X_1=\{x_1,\ldots, \blue{x_r}\}$ must be a phylogenetic tree on $X_1$ 
	\blue{where, for all $3\leq j \leq k$, we put $x_j=x(m_j)$.}
	Subject to potentially
	having to relabel the leaves of $T$, we may assume that $\{x_1,x_2\}$ is a cherry in $T$. 
	Since $\alpha=0$ the directed acyclic graph \blue{$Q'$} obtained from $Q$ by deleting $x_1$ and its incoming
	arc (suppressing resulting vertices of indegree
	and outdegree one) and renaming $x_{i+1}$ by \blue{$x(m_i')$}, for all $1\leq i\leq n-1$, is a phylogenetic 
	network on \blue{$\{x(m_1'),\ldots, x(m'_{n-1})\}$}.	Clearly, \blue{$Q'$} realizes $\vec{m'}$ \blue{since $Q$ realizes $\vec{m}$}. 
	By assumption on $\vec{m}$ it follows that \blue{ $N'$ is an attainment of
		$\vec{m'}$. Hence, $h(N')\leq h(Q')$}. 	Since $N$ is obtained
	from \blue{$N'$} by \blue{executing Line~\ref{alpha=0} in Algorithm~\ref{alg:Nmalgorithm}}
	%
	 it follows that $h(Q)<h(N)= \blue{h(N')\leq h(Q')}=h(Q)$ because $T$ is a tree; a contradiction. Consequently, \blue{$N$} must attain
	$\vec{m}$ in this case.\\

	\noindent {\bf Case (b):} Assume that \blue{$\alpha\leq m_2$. Let $j$, $x_1$, and $x_2$ be as in Line~\ref{cond-alpha-smaller} in Algorithm~\ref{alg:Nmalgorithm}.} We start with analyzing the structure of $Q$ with
		regards to $x_1$ and $x_2$. To this end, note first that $m_2\geq 2$ must hold
		since otherwise $\vec{m}$ is simple and the theorem follows in view of our observation at the beginning of the proof. 
			
By assumption on $Q$,
			there must exist a hybrid vertex $h$ on $P$ such that there is a directed path $P_h$ from $h$ to $x_2$
			because $m_2\geq 2$.  Without loss of generality,
			we may assume that $h$ is such that every vertex on 
			$P_h$ other than $h$ is either a tree vertex or a leaf
			of $Q$. Let $t$ be the last vertex on $P$ that is also
			contained in $P_h$.  
			
			We next transform $Q$ into a new phylogenetic network $Q''$ that is an attainment of $\vec{m'}$ (see Figure~\ref{fig:theo-illustrate}
				\begin{figure}[h!]
					\centering
					\includegraphics[scale=0.42]{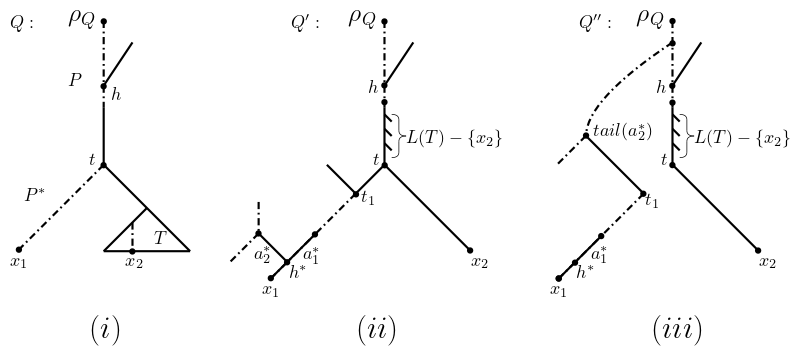}\
					\caption{The transformation of $Q$ (i) into
						the phylogenetic networks $Q'$ (ii) and
						$Q''$ (iii) as described in Case~(b)
					of 	Theorem~\ref{theo:min} for \blue{$p_1\not=p_2$}. In each case, the dashed lines indicate paths. Note that in (iii) the dashed line could also start at $\rho_Q$.
						\label{fig:theo-illustrate}}
				\end{figure}
				for an illustration). To do this, note first that since $m_2\not=m_1$ there must exist a hybrid vertex on $P$ below $t$. We  modify $Q$ as follows to obtain a further attainment $Q'$
			of $\vec{m}$. If $t$ is the parent of $x_2$ then
			$Q'$ is $Q$. So assume that $t$ is not the parent of 
			$x_2$. Then we delete the subtree $T$ of $Q$ that is rooted at the child of $t$ not contained in $P$. Note that $T$ must have at least two leaves. Next, we subdivide the incoming arc of $t$ by $|L(T)|-1$
			subdivision vertices. To each created subdivision vertex we add \blue{an} arc and bijectively label the
			heads of these arcs by the elements in $L(T)-\{x_2\}$.
			Next, we add \blue{an} arc to $t$ and label its head by $x_2$ so that $t$ is now the parent of $x_2$.
			\blue{By construction,} $Q'$ is a phylogenetic network
			on $X$ that attains $\vec{m}$ because $h(Q)=h(Q')$.
			
			Let $h^*$ be a hybrid vertex on the subpath $P^*$ of
			$P$ from $t$ to $x_1$ so that no vertex strictly 
			below $h^*$ is a hybrid vertex of $Q'$. Let $a_1^*$ denote the incoming arc of $h^*$ that lies on $P^*$. 
			Furthermore, let $a_2^*$ denote the incoming arc of $h^*$ that does not lie on $P^*$.
			For $i=1,2$, let $p_i$ denote the tail of $a_i^*$. Note that $p_1=p_2$ might hold. Also note that the assumptions on $Q$ imply that  $p_1$ must be below $t$.  Finally, note that $p_1$ must be a hybrid vertex unless	$p_1=p_2$. 
		
		We claim that if $p_1\not=p_2$ then any vertex $v$ on $P^*$ other than $t$ and $x_1$ must be a hybrid vertex.   Assume for contradiction that there exists a vertex $v\not\in\{t,x_1\}$ on $P^*$ that is a tree vertex. 
			We show first that $p_2$ must also be below $t$.
			\blue{Since all hybrid vertices of $Q$ lie on $P$, it follows that}, 
			$v$ contributes at least $2m_2$ to the number of directed paths from $\rho$ to $x_1$ as $m_2$ 
			is the number of directed paths from $\rho$
			to $x_2$ and therefore, also from $\rho$ to $t$.
			Since $h_1^*$ contributes at least one further directed path from $\rho$ to $x_1$ in case $p_2$ is not below $t$, it follows that $m_1\geq \beta+\blue{2}m_2$ for some $\beta\geq 1$. Hence,
			$m_2\geq\alpha=m_1-m_2 \geq\beta +\blue{2}m_2-m_2\geq m_2$ because $\beta\geq 1$.
			Thus, $m_2=\beta+m_2$; a contradiction as $\beta\geq 1$. Hence, $p_2$ must also be below $t$, as required.
			
			\blue{We next show that $p_2$ must be a vertex on $P^*$. Indeed,}
			if $p_2$ were not a vertex of $P^*$ then it
			cannot be a hybrid vertex in view of our assumptions on $Q$. Thus, $p_2$  must be a tree vertex in this case. Since $p_1\not=p_2$
			we obtain a contradiction as the choice of $h^*$
			implies that $h^*$ is the parent of $x_1$. Thus,
			$p_2$ must be a vertex of $P^*$, \blue{ as required. Since $p_2$ is a tree vertex it  
			contributes at least $2m_2$ directed paths from $\rho$ to
			$x_1$. Since $p_1$ contributes at least a further $m_2$ directed paths from $\rho$ to $x_1$, we obtain
			a contradiction using similar arguments as before.} Thus any
vertex on $P^*$ other than $t$ and $x_1$ must be a hybrid vertex in case $p_1\not=p_2$, as claimed.

We claim that if $p_1=p_2$ then $P^*$ has precisely 4 vertices and there exists two arcs from $p_1$ to $h^*$. To see this claim, note that $p_1$ contributes at least $2m_2$
	directed paths from $\rho$ to $x_1$ \blue{because it is a tree vertex}. If there existed
	a vertex $v$ on $P^*$ distinct from
	$x_1$, $h^*$, $p_1$, $t$ then $v$ would contribute at least $m_2$ further directed paths from
	$\rho$ to $x_1$. Thus, we have again at least $3m_2$ directed paths from $\rho$ to $x_1$. Similar arguments as in the previous claim yield again a contradiction.  
	By the choice of $h^*$ it follows that $t$, $p_1$, $h^*$ and $x_1$ are the only vertices on $P^*$. Since $p_1$
	and $p_2$ are the parents of $h^*$ and $p_1=p_2$, it follows that
	there are two parallel arcs from $p_1$ to $h^*$. This concludes the proof of 
	\blue{our second} claim.

Bearing in mind the previous two claims, we 
	next transform $Q'$ into a new phylogenetic network $Q''$ on $X$ as follows. If $p_1\not=p_2$ then we first delete $a_2^*$ from $Q'$ and add an arc from 
				$p_2$ to the child $t_1$ 
				of $t$ on $P^*$. 
				Next, we remove the arc $(t,t_1)$ and suppress $h^*$ and $t$ as they are now vertices with
				indegree one and outdegree one. The resulting directed acyclic graph is
				$Q''$. \blue{By construction,} $Q''$ is \blue{clearly} a phylogenetic network on $X$. \blue{Furthermore, the} construction \blue{combined with our}
				 two claims, \blue{implies} that 
				$Q''$ realizes $\vec{m}'$ because the arc $(t,t_1)$ contributes $m_2$ directed paths from $\rho$ to $x_1$ in $Q$ and therefore also in $Q'$. By construction, 
				$h(Q'')=h(Q')-1=h(Q)-1$. Furthermore, $h(N)=h(\blue{N'})+1$ by the construction of $N$ from \blue{$N'$}. By the minimality of $h(Q)$ 
				\blue{and the choice of $\vec{m}$,} it follows that $h(Q)<h(N)=h(\blue{N'})+1\leq h(Q'')+1=h(Q)$;
				a contradiction. This concludes the proof of the theorem in case \blue{$p_1\not=p_2$.}
			
			If $p_1=p_2$ then we delete one of the
				two parallel arcs from $p_1$ to $h^*$ and suppress $p_1$ and $h^*$ as this has rendered them vertices of indegree one and outdegree one. The resulting
				directed acyclic graph is $Q''$ \blue{in this case}.
				As before, $Q''$ is a phylogenetic network
				that, in view of our second claim, realizes \blue{$\vec{m'}$}. Similar arguments
				as in the case that $p_1\not=p_2$ yield again a contradiction. This concludes the 
				proof of the theorem in this case,
				and therefore, the proof of the theorem.
\end{proof}

To illustrate Theorem~\ref{theo:min}, note that the ploidy profile $\vec{m}=(12, 6, 6, 5)$ in
	Figure~\ref{fig:intro-fig} satisfies the assumptions of Theorem~\ref{theo:min}. Consequently,
	the phylogenetic network $N(\vec{m})$
	depicted in that figure is an attainment of $\vec{m}$.

As the example depicted in Figure~\ref{fig:barrednetwork} indicates, the assumption that
\blue{Line~\ref{alpha-larger} in  Algorithm~\ref{alg:Nmalgorithm} is not 
	executed} is necessary for Theorem~\ref{theo:min} to hold. 
	In fact, if $\vec{m}$ is a ploidy profile such that
	$N(\vec{m})$ contains the subgraph highlighted by the dashed
	rectangle in the network in Figure~\ref{fig:barrednetwork},
	then $N(\vec{m})$ can in general not  be an attainment of $\vec{m}$.
	\begin{figure}[h]
		\centering
		\includegraphics[scale=0.3]{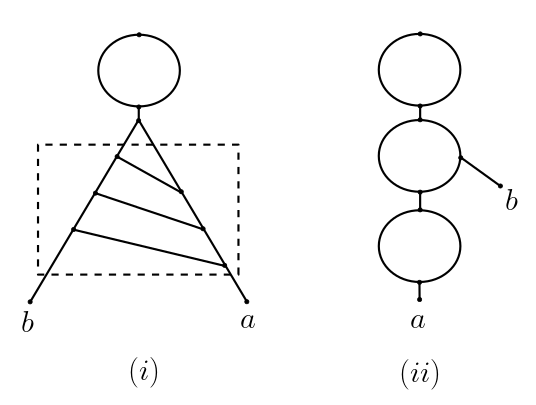}
		\caption{\label{fig:barrednetwork}
			(i) \blue{The} phylogenetic network $N(\vec{m})$ for the ploidy profile $\vec{m}=(8,2)$ on $X=\{a,b\}$ obtained via Algorithms~\ref{alg:simplification-seq} and \ref{alg:Nmalgorithm}.
			(ii) A phylogenetic network on $X$ that attains $\vec{m}$ and has fewer hybrid vertices than $N(\vec{m})$.
		}	
	\end{figure}

Theorem~\ref{theo:min} and Case~(b) in its proof combined with Theorem~\ref{Bm-vertex-count} and Proposition~\ref{prop:upper-bound} implies our next result
\orange{since $l-1$ additional hybrid vertices are inserted into  $B(i_1)$ to obtain $B(\vec{m})$ 
where $\vec{m}$ is a simple ploidy profile and $(i_1,\ldots, i_l)$, $l\geq 1$, is the binary 
representation of the first component of $\vec{m}$.}
To state it we require a further definition. Let $\vec{m},\vec{m}_1,\ldots,\vec{m}_i=(m_{1,i},\ldots, m_{p_i,i}),\ldots, \vec{m}_t$  denote the simplification
sequence of a ploidy profile
$\vec{m}$. Then we denote by $c(s(\vec{m}))$ the number of steps in $\sigma(\vec{m})$, for which $m_{1,i}>m_{2,i}$ holds where $0\leq i\leq s(\vec{m})$ and $p_i\geq 1$.

\begin{corollary}\label{cor:upper-bound}	
	Suppose $\vec{m}$ is a ploidy profile 
	\blue{such that Line~\ref{simp:alpha-larger} in Algorithm~\ref{alg:simplification-seq} is not executed when constructing
	$\sigma(\vec{m})$.} Then  $h(\vec{m})= h(\vec{m}_t)+c(s(\vec{m}))$.
	If $B(\vec{m}_t)$ is an attainment of $\vec{m}_t$ and
	\blue{$(i_1,\ldots,i_l)$}
	is the binary representation of \blue{the first component of} $\vec{m}_t$, some 
	\blue{$l\geq 1$},
	then $h(\vec{m})= i_1+\orange{l-1} +c(s(\vec{m}))$.
\end{corollary}

\section{A \orange{Viola} dataset}
\label{sec:viola}
In this section, we turn our attention to computing the hybrid number of the ploidy profile 
of a \orange{Viola}  dataset 
that appeared in more general form in
\cite{MJDBBBO12}. Denoting that dataset by $X$, the authors of \cite{MJDBBBO12} constructed a MUL-tree $ M$ on $X$ and then used the PADRE software \cite{HOLM06} to derive a phylogenetic network $N$ to help them shed light on the evolutionary past of their \orange{Viola}  species \blue{\cite[Figure 4]{MJDBBBO12}}. We depict a simplified network $N'$ representing that past in Figure~\ref{fig:viola}(i) the only
difference being that we have removed species that are not below a hybrid vertex of $N$ as they do not contribute to the number of hybrid vertices of $N$. If more than one species were below a hybrid vertex of $N$, then we have also randomly removed all but one of them thereby ensuring that the hybrid vertex is still present in $N'$.
\blue{The resulting simplified dataset comprises the taxa
	$x_1=${\em V.langsdorffii}, $x_2=${\em V.tracheliifolia}, $x_3$= {\em V.grahamii},
	$x_4=${\em V.721palustris}, $x_5=${\em V.blanda}, $x_6=${\em V.933palustris},
	$x_7=${\em V.glabella}, $x_8=${\em V.macloskeyi}, $x_9=${\em V.repens}
	$x_{10}=${\em V.verecunda}, $x_{11}=${\em Viola}, and $x_{12}=${\em Rubellium} (see \cite{HM22} for more details on the simplified dataset)}. The labels of the internal vertices of $N'$ represent the ploidy number of the ancestral species represented by that vertex where we canonically extend the concept of a ploidy profile to the interior vertices of a phylogenetic network. \blue{By counting directed paths from the root to each leaf, it}  is easy to check, 
\blue{$h(N')=9$}. 

By taking directed paths from the root to the leaves of $N'$, we obtain the ploidy profile \blue{$\vec{m} = (9,7,7,4,4,4,2,2,2,2,2,1)$} on $X$. Note, since the root is diploid (labelled $2\times$), multiplying each component of $\vec{m}$ by two results in the 
ploidy numbers induced by the hybrid vertices in the network. The simplification sequence for $\vec{m}$ 
contains twelve elements and $\vec{m}_t=(2,1,1,1)$. Since 
an attainment of $\vec{m}_t$ must have \blue{one hybrid
vertex and $D(\vec{m}_t)$ are equal $B(\vec{m}_t)$ and have  one hybrid vertex each, it
follows that $B (\vec{m}_t)$ is an attainment for $\vec{m}_t$.
The phylogenetic network $N(\vec{m})$ obtained by initializing Algorithm~\ref{alg:Nmalgorithm}
with $B(\vec{m}_t)$ is depicted in Figure~\ref{fig:viola}(ii). Since at no stage in the
construction of $N(\vec{m})$ Line~\ref{alpha-larger} of that algorithm  is executed, it follows by  
Theorem~\ref{theo:min} that $N(\vec{m})$} is an attainment of $\vec{m}$. \blue{Counting
again directed paths from the root to each leaf, it} is easy to check \blue{that} 
$N(\vec{m})$ has five hybrid vertices implying that $h(\vec{m})=5$. To compute \blue{the} hybrid 
number of a ploidy profile whose components are not too large \blue{and, thererfore, we can find an attainment of its terminal element}, we refer the interested 
reader to our R-function `ploidy profile hybrid number bound (PPHNB)' which is obtainable from \blue{\cite{G22}}.

\begin{figure}
	\centering
	\includegraphics[scale= 0.38]{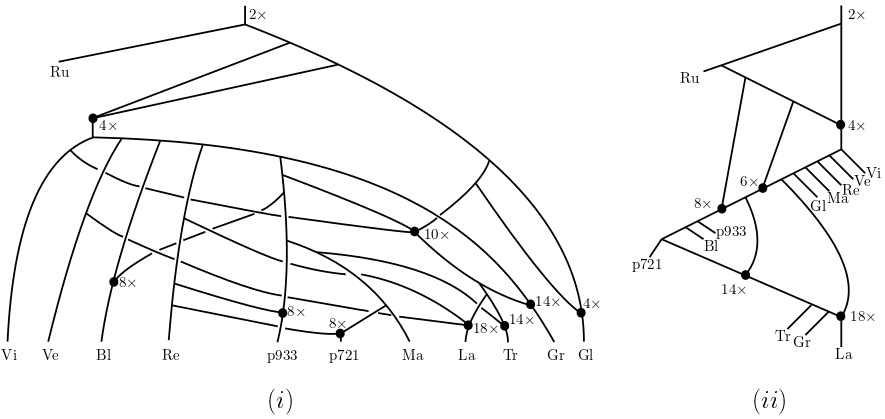}
	\caption{A phylogenetic network on leaf set $X = \{$\blue{\em V.langsdorffii, V.tracheliifolia, V.grahamii, V.721palustris, V.blanda, V.933palustris, V.glabella,  V.macloskeyi, 
	V.repens, V.verecunda,} Viola, Rubellium$\}$ adapted from a more general phylogenetic network that \blue{appeared as Figure 4} in \cite{MJDBBBO12}. Hybrid \blue{vertices} are indicated with a filled circle and labelled \blue{by}
		their corresponding ploidy number i.\,e.\,\blue{the} number of directed paths from the root to the vertex times two because the root is assumed to be diploid. Leaves are labelled by the first two
		characters of their names (omitting {\em 'V.'}, where applicable). }
	\label{fig:viola}
\end{figure}

\section{Discussion}
\label{sec:discussion}
Motivated by the signal left behind by \blue{polyploidization}, we have introduced 
and studied the problem of \green{computing} the hybrid number $h(\vec{m})$ of a 
ploidy profile $\vec{m}$. \blue{Our} arguments
apply, \blue{however,} to any type of dataset that induces a multiplicity vector. 
Although stated within a phylogenetics context, the underlying optimization problem is, at its heart, a natural mathematical problem: ``Given a 
	multiplicity vector $\vec{m}$ find a rooted, leaf-labelled, directed acyclic graph $G$  so that $\vec{m}$
	is the path-multiplicity vector of $G$ and the cyclomatic number of $G$ is minimum''.
	Our results might therefore be also of relevance beyond phylogenetics. 
	
Using the 
framework of a phylogenetic network, 
we provide a construction of a phylogenetic 
network $N(\vec{m})$ that is guaranteed to 
attain \blue{a ploidy profile $\vec{m}$}  for a large class of ploidy profiles provided the construction of $N(\vec{m})$ is initialized with an attainment $\mathcal A(\vec{m}_t)$ of the terminal element $\vec{m}_t$ of \blue{the} simplification sequence $\sigma(\vec{m})$ associated to $\vec{m}$. \blue{Members of that class include 
	the ploidy 
	profiles described in Proposition~\ref{prop-simp-length}(ii). As a consequence,
we obtain} an exact formula for \blue{the hybrid number
of} $\vec{m}$
and also the size of the vertex set of $N(\vec{m})$
in terms of the length $s(\vec{m})$ of  $\sigma(\vec{m})$ and the number $a(\vec{m}_t)$ of vertices of $\mathcal A(\vec{m}_t)$
for the members of \blue{our} class. In case the ploidy numbers 
that make up $\vec{m}$ are not too large, both $c(s(\vec{m}))$ and $a(\vec{m}_t)$ can be  computed easily \blue{by computing $\sigma(\vec{m}$)  to obtain $c(s(\vec{m}))$
and using, for example, an exhaustive search for $a(\vec{m}_t$)}. Having said this, we also present an infinite family of ploidy profiles $\vec{m}$ for
which $\sigma(\vec{m})$ grows exponentially. 
\green{Motivated by this,} we provide 
a bound for $h(\vec{m})$ and show that that bound is sharp
for certain types of ploidy profiles. 
\green{To help demonstrate the applicability of our approach, we compute the hybrid number of a simplified version of a \orange{Viola}  dataset that appeared in more general form in \cite{MJDBBBO12}}. Our result suggests that
the authors of \cite{MJDBBBO12} potentially overestimate the
number of polyploidization events that gave rise to their dataset.


Despite these
encouraging results, numerous questions that might merit further research remain. These include 
\green{``What can be said about $h(\vec{m})$ if the ploidy profile $\vec{m}$ is not a member of our class?'', and ``Can we shed more 
light on the length of $\sigma(\vec{m})$
and also into attainments of the terminal element 
of $\sigma(\vec{m})$?''. Looking a little bit further afield,}
it might also be of interest to explore the relationship
between so called accumulation phylogenies introduced in 
\cite{BS06} and ploidy profiles and also  
the relationship between ploidy profiles and ancestral profiles introduced in
\cite{ESS19}.\\

\green{
	\noindent{\bf Acknowledgment}
	We thank the anonymous referees for their constructive comments to improve \orange{earlier versions} of the paper.
	 }

 \bibliographystyle{plain}
 \bibliography{bibliography2021-01-04}

\end{document}